\documentclass{article}
\usepackage{euscript,amsmath,amssymb,amsfonts,graphicx,bm,amsthm,mathtools}

\usepackage{cite} 

  \usepackage{paralist}
  \usepackage{graphics} 
\usepackage[caption=false]{subfig}
\usepackage{soul}

\usepackage{latexsym,amssymb,enumerate,amsmath,amsthm} 
  \usepackage{graphicx} 
  \usepackage{tikz}
     \usepackage[colorlinks=false]{hyperref}
 \hypersetup{urlcolor=blue, citecolor=red}
\usepackage{latexsym,amssymb,enumerate,amsmath} 
\usepackage{pgfplots}
\usepackage{soul,xcolor}

\newtheorem{theorem}{Theorem}

\newtheorem{corollary}[theorem]{Corollary}
\theoremstyle{plain}
\theoremstyle{remark}
\newtheorem{remark}[theorem]{Remark}
\theoremstyle{definition}
\newtheorem{definition*}{Definition}

\newcommand{\dd}{\text{d}}

\def\E{\mathbb{E}}
\def\P{\mathbb{P}}
\def\R{\mathbb{R}}
\def\Z{\mathbb{Z}}

\def\thalf{t_{\textup{half}}}
\def\A{A}
\def\cv{c_{\textup{v}}}
\def\single{\textup{single}}
\def\double{\textup{double}}
\def\triple{\textup{triple}}
\def\all{\textup{all}}

\def\p{p}
\def\dev{\Delta}
\def\ka{k_{\textup{a}}}
\def\ke{k_{\textup{e}}}
\def\perf{\textup{perf}}
\def\AUC{\textup{AUC}}
\def\blue{\color{black}}

\begin{document}


\title{What should patients do if they miss a dose of medication? A theoretical approach}


\author{Elijah D. Counterman\thanks{Department of Mathematics, University of Utah, Salt Lake City, UT 84112 USA.} \and Sean D. Lawley\thanks{Department of Mathematics, University of Utah, Salt Lake City, UT 84112 USA (\texttt{lawley@math.utah.edu}).}
}
\date{\today}
\maketitle

\begin{abstract}
Medication adherence is a major problem for patients with chronic diseases that require long term pharmacotherapy. Many unanswered questions surround adherence, including how adherence rates translate into treatment efficacy and how missed doses of medication should be handled. To address these questions, we formulate and analyze a mathematical model of the drug {\blue concentration} in a patient with imperfect adherence. We find exact formulas for drug {\blue concentration} statistics, including the mean, the coefficient of variation, and the deviation from perfect adherence. We determine how adherence rates translate into drug {\blue{concentration}s}, and how this depends on the drug half-life, the dosing interval, and how missed doses are handled. While clinical recommendations {\blue require extensive validation and} should depend on drug and patient specifics, as a general principle {\blue our theory suggests} that nonadherence is best mitigated by taking double doses following missed doses if the drug has a long half-life. This conclusion contradicts some existing recommendations that cite long drug half-lives as the reason to avoid a double dose after a missed dose. Furthermore, we show that a patient who takes double doses after missed doses can have at most only slightly more drug in their body than a perfectly adherent patient if the drug half-life is long. We also investigate other ways of handling missed doses, including taking an extra fractional dose following a missed dose. We discuss our results in the context of hypothyroid patients taking levothyroxine.
\end{abstract}

\section{Introduction}

Adherence to medication{\blue s is the process by which patients take their medications as prescribed} \cite{vrijens2012}. It is well-documented that nonadherence is a major problem, resulting in over 100,000 preventable deaths and \$100 billion in preventable health care costs per year in the United States alone \cite{osterberg2005}. In fact, the World Health Organization noted that ``increasing the effectiveness of adherence interventions may have a far greater impact on the health of the population than any improvement in specific medical treatments'' \cite{who2003, haynes2002}. Nonadherence is especially prevalent and problematic in patients with chronic diseases that require long term pharmacotherapy \cite{burnier2019}. As former US surgeon general C Everett Koop famously observed, ``Drugs don't work in patients who don't take them'' \cite{lindenfeld2017}.

{\blue Medication adherence has been divided into the three phases of initiation, implementation, and discontinuation \cite{vrijens2012}. Initiation and discontinuation refer to a patient starting and stopping a regimen as prescribed (and the term ``persistence'' describes the time from initiation to discontinuation \cite{vrijens2012, geest2018}). In this paper, we focus on implementation, which is the extent to which a patient's actual dosing follows the prescribed dosing regimen \cite{vrijens2012}.} 

Many outstanding questions surround {\blue the implementation phase of} adherence and how it relates to therapeutic outcomes. Adherence is often reported as the percentage of doses of medication actually taken by the patient over a specified
time \cite{osterberg2005}. How does an adherence percentage $p$ translate into treatment efficacy? How much worse is, for example,  $p=70\%$ compared to $p=85\%$? How much adherence is needed for full treatment benefits? How can clinicians increase patient adherence? Are there protocols to increase treatment benefits in spite of poor adherence?

While the causes of nonadherence vary, a significant portion of nonadherence stems from patients simply forgetting to take their medication \cite{spilker1991, barfod2006}. What should a patient do if they miss a dose of medication? Although patients commonly ask this question, they often do not receive adequate instructions for what to do when a dose is missed \cite{howard1999, albassam2020, gilbert2002}. 

To address these questions, we formulate and analyze a mathematical model of the drug {\blue concentration} in a patient with imperfect adherence. Mathematical modeling is especially well-suited to investigate these questions, given the ethics of clinical trials that force patients to miss doses of medication. To model imperfect adherence, we assume that the patient takes their medication at only a given percentage of the prescribed dosing times. Doses are missed at random, and thus the drug {\blue concentration} in the body is random. {\blue For simplicity, we assume that the patient misses each dose with a fixed probability $1-p$, independent of their prior behavior.} We find exact mathematical formulas for statistics of this model, including the average drug {\blue concentration}, the drug {\blue concentration} coefficient of variation, and how the drug {\blue concentration} deviates from a patient with perfect adherence. These statistics are obtained as explicit functions of the adherence percentage, the drug half-life, and the prescribed dosing interval (i.e.\ the time between scheduled doses). Furthermore, we determine how these statistics depend on how the patient handles missed doses, including the case that they skip missed doses and the case that they take double doses following missed doses.

From a mathematical standpoint, the random variables that model the drug {\blue concentration} in our model generalize infinite Bernoulli convolutions \cite{peres2000, solomyak1995, peres1998, escribano2003, hu2008}. The study of infinite Bernoulli convolutions has a rich history in the pure mathematics literature, dating back to Erd\H{o}s and others in the 1930s \cite{jessen1935, kershner1935, erdos1939}. Infinite Bernoulli convolutions typically have very irregular distributions, including singular distributions supported on a Cantor set \cite{kershner1935}. Infinite Bernoulli convolutions also arose in the pharmacokinetic models in \cite{levy2013, fermin2017}. Our analysis of the generalized infinite Bernoulli convolutions that arise in our model relies on the theory of random pullback attractors \cite{Crauel01, Mattingly99, Schmalfuss96, lawley15sima, lawley2019hhg}.

From the standpoint of pharmacology, there are {\blue several results} of our analysis. First, we provide quantitative estimates of how an adherence percentage $p$ translates into statistics of drug {\blue concentration}s in the body, and how these statistics depend on the drug half-life $\thalf$, the dosing interval ${\tau}$, and how missed doses are handled. Further, these estimates show how the effects of nonadherence can be lessened by drugs with half-lives that are long compared to the dosing interval, i.e.\ $\thalf\gg {\tau}$. While clinical recommendations {\blue require extensive validation and} should depend on drug and patient specifics, as a general principle {\blue our theory suggests} that the effects of nonadherence are best mitigated by taking double doses following missed doses if $\thalf\gg {\tau}$, whereas missed doses should be skipped if $\thalf\ll {\tau}$. This conclusion contradicts some existing recommendations that cite long drug half-lives as the reason to avoid a double dose after a missed dose {\blue(for example, see recommendations for perampanel \cite{albassam2020} and valproate \cite{gilbert2002}), as well as the general recommendation that double doses should not be taken to compensate for missed doses \cite{siwale2016}}. Since double doses are sometimes avoided due to concern that they may cause toxic drug {\blue concentration}s, we provide an upper bound for the highest possible drug {\blue concentration} in the body. We find that a patient who takes double doses after missed doses can  {\blue have at most only a slightly higher drug concentration (and exposure) than} a perfectly adherent patient if $\thalf\gg {\tau}$. We also investigate other ways of handling missed doses, including taking an extra half dose following a missed dose, which we find is most appropriate when $\thalf\approx {\tau}$.

The rest of the paper is organized as follows. We formulate and analyze the mathematical model in the Methods section {\blue(details of the mathematical analysis are in the Appendix)}. In the Results section, we explore the pharmacological implications of the mathematical analysis. Since these pharmacological implications depend on rather complicated mathematics, we also provide an intuitive explanation for our results in this section. The Discussion section concludes by describing related work, model limitations, and future directions. We also discuss our results in the context of hypothyroid patients taking levothyroxine. The Appendix collects some technical points and the proofs of the theorems.

\section{Methods}

\subsection{Mathematical model}\label{model}

Our model builds on the classical pharmacokinetic model of extravascular (oral) administration in a single compartment with first order kinetics \cite{gibaldi1982, bauer2015}. In the standard model, the drug concentration, $c_{0}$, in the body at time $s>0$ satisfies the ordinary differential equation (ODE),
\begin{align}\label{code}
\frac{\dd c_{0}}{\dd s}
=\ka\frac{g}{V}-\ke c_{0},
\end{align}
where $\ka$ and $\ke$ are the respective rates of absorption and elimination, $V$ is the volume of distribution, and $g$ is the drug amount at the absorption site. The amount $g$ satisfies the ODE,
\begin{align}\label{gode}
\frac{\dd g}{\dd s}
=-\ka g+I(s),
\end{align}
where $I(s)$ describes the drug input.

For most drugs administered extravascularly in conventional dosage forms, the absorption rate is much larger than the elimination rate, meaning $\ka\gg\ke$ (see \cite{gibaldi1982, peletier2017, ma2018, fillastre1987, leroy1990, strandgaarden1999}). In this parameter regime, the solution of \eqref{code} is well-approximated by the solution to
\begin{align}\label{simpler}
\frac{\dd c}{\dd s}
=\frac{I(s)}{V}-\ke c,
\end{align}
which is the standard model for intravascular administration with first order elimination. In this paper, we assume $\ka\gg\ke$ and thus consider the simpler model in \eqref{simpler} rather than the system in \eqref{code}-\eqref{gode}.

\subsubsection{Perfect adherence}

Suppose a patient is instructed to take a dose of size $D>0$ at regular time intervals of length $\tau>0$ beginning at time $0$. If the patient has perfect adherence, then the drug input is
\begin{align}\label{iperf}
I^{\perf}(s)
=DF\sum_{n\ge0}\delta(t-n\tau),
\end{align}
where $F\in(0,1]$ is the bioavailability fraction and $\delta$ denotes the Dirac delta function. Solving \eqref{simpler}-\eqref{iperf} yields the following well-known formula for the drug concentration at time $s\ge0$ in the perfectly adherent patient \cite{bauer2015},
\begin{align}\label{cstart}
c^{\perf}(s)
:=
\frac{DF}{V}\sum_{n=0}^{N(s)}e^{-\ke(s-n\tau)},
\end{align}
where $N(s)+1$ is the number of dosing times elapsed by time $s$,
\begin{align*}
N(s)
:=\sup\{n\ge0:n\le s/\tau\}.
\end{align*}
If
\begin{align*}
t=s-N(s)\tau\in[0,\tau)
\end{align*}
denotes the time elapsed since the most recent dosing time, then \eqref{cstart} can be written as
\begin{align}\label{cform}
c^{\perf}(s)
=
\alpha^{t/\tau}\frac{DF}{V}\sum_{n=0}^{N(s)}\alpha^{n},
\end{align}
where we have defined the dimensionless constant
\begin{align}\label{alphabeta}
\alpha:=e^{-\ke{\tau}}\in(0,1),
\end{align}
which is the fraction of a dose that remains in the body after one dosing interval.

If the patient continues their perfect adherence for a long time, then it is easy to see from the form in \eqref{cform} that the drug concentration approaches the following function,
\begin{align}\label{perfper}
c^{\perf}(N\tau+t)
\to C^{\perf}(t)
:=\alpha^{t/\tau}\frac{DF}{V}A^{\perf}\quad\text{as }N\to\infty,
\end{align}
where $t\in[0,\tau)$ is the time since the last dose and
\begin{align*}
A^{\perf}:=
\sum_{n\ge0}\alpha^{n}
=\frac{1}{1-\alpha}.
\end{align*}
In pharmacokinetics, it is common to measure the drug exposure over a single dosing interval by the so-called ``area under the curve,'' which for this case of perfect adherence is
\begin{align}\label{aucperf}
\AUC^{\perf}
:=\int_{0}^{\tau}C^{\perf}(t)\,\dd t
=\frac{DF}{V}\frac{1}{\ke}.
\end{align}

\subsubsection{Nonadherence}

To model patient nonadherence, we suppose that the patient occasionally misses a dose. Specifically, at each dosing time, the patient ``remembers'' to take their medication with probability $p\in(0,1)$, and the patient ``forgets'' with probability $1-p$. Mathematically, let $\{\xi_{n}\}_{n}$ be a sequence of independent and identically distributed (iid) Bernoulli random variables with parameter $p$, meaning
\begin{align}\label{xin}
\begin{split}
\xi_{n}
=\begin{cases}
1 & \text{with probability }p,\\
0 & \text{with probability }1-p.
\end{cases}
\end{split}
\end{align}
Hence, $\xi_{n}=1$ means that the patient takes their medication at the $n$th dosing time. {\blue We emphasize that $\{\xi_{n}\}_{n}$ is a sequence of independent random variables, which means that the patient misses doses independently of their prior behavior.}

If $Df_{n}\ge0$ denotes the amount taken at the $n$th dosing time, then the drug input is
\begin{align}\label{iimp}
I(s)
=DF\sum_{n\ge0}\delta(t-n\tau)f_{n},
\end{align}
and solving \eqref{simpler} with $I(s)$ in \eqref{iimp} yields the drug concentration in the patient,
\begin{align}\label{CC0}
c(s)
=
\alpha^{t/\tau}\frac{DF}{V}\sum_{n=0}^{N(s)}\alpha^{N(s)-n}f_{n}.
\end{align}
Notice that \eqref{CC0} reduces to \eqref{cform} if $f_{n}=1$ for all $n$. We take
\begin{align*}
f_{n}=0,\quad\text{if }\xi_{n}=0,
\end{align*}
which means the patient does not take any medication when they forget. However, we allow for the possibility that 
\begin{align*}
f_{n}> 1,\quad \text{if }\xi_{n}=1,
\end{align*}
which means that the patient may take more than a single dose to make up for prior missed doses. In general, we allow $f_{n}$ to be a function of the history $\{\xi_{i}\}_{i\le n}$, and we refer to a choice of $f_{n}$ as a ``dosing protocol.''

The simplest dosing protocol is for the patient to merely take a single dose if they remember, which means
\begin{align}
\label{fsingle}
f_{n}^{\textup{single}}
:=
\begin{cases}
0 & \text{if }\xi_{n}=0,\\
1 &  \text{if }\xi_{n}=1.
\end{cases}
\end{align}
We refer to \eqref{fsingle} as the ``single dose'' protocol. Another common dosing protocol is for the patient to take a double dose to make up for a missed dose at the prior dosing time, which means
\begin{align}
\begin{split}\label{fdouble}
f_{n}^{\textup{double}}
:=\begin{cases}
0 & \text{if }\xi_{n}=0,\\
1 & \text{if }\xi_{n}=1,\,\xi_{n-1}=1,\\
2 & \text{if }\xi_{n}=1,\,\xi_{n-1}=0. 
\end{cases}
\end{split}
\end{align}
We refer to \eqref{fdouble} as the ``double dose'' protocol. Notice that in the double dose protocol, the patient never takes more than two doses at a time, even if they missed more than one previous dose. Our analysis below covers other dosing protocols, but we are primarily interested in comparing the single dose and double dose protocols in \eqref{fsingle}-\eqref{fdouble}. As a technical aside, we are ultimately interested in the large time behavior of $c(s)$ in \eqref{CC0}, and thus the values of $f_{n}$ in \eqref{CC0} for small $n$ are irrelevant. In particular, the fact that the definition of $f_{0}^{\textup{double}}$ in \eqref{fdouble} depends on $\xi_{-1}$ is immaterial.

\begin{figure}[t]
\centering
\includegraphics[width=\linewidth]{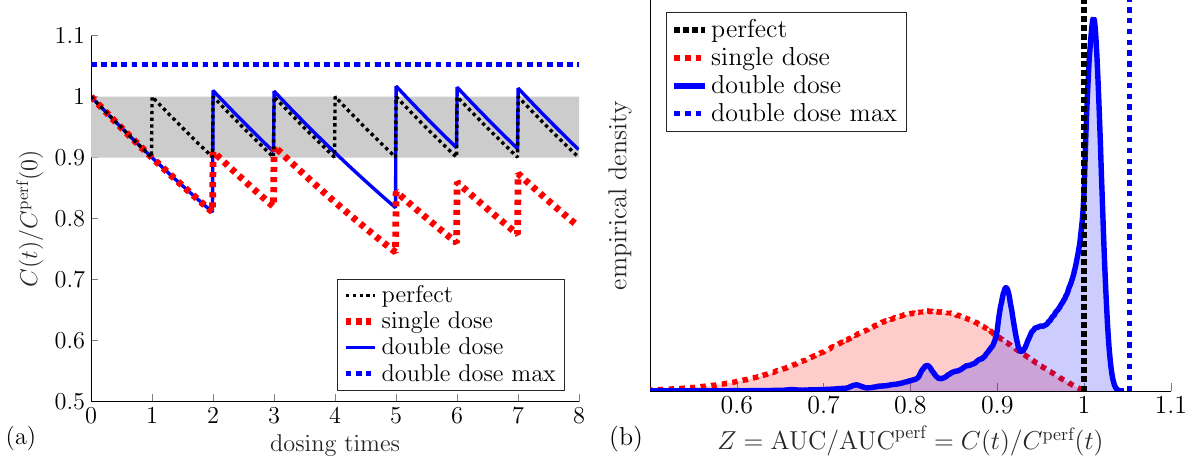}
\caption{(a) The black dotted curve depicts how the drug concentration in a perfectly adherent patient evolves in time, and the gray shaded region is the area between the peaks and troughs of the perfectly adherent patient. The red dashed curve and blue solid curve describe patients with imperfect adherence following the single dose and double dose protocols, respectively. The blue dashed line depicts the largest possible drug concentration for the double dose protocol (see {\blue\eqref{largeste}}). (b) The distribution of the relative drug {\blue concentration} $Z=\AUC/\AUC^{\perf}=C(t)/C^{\perf}(t)$ for the single dose protocol (red) and the double dose protocol (blue) obtained from stochastic simulations. The black dotted vertical line at $Z=1$ describes the perfectly adherent patient, and the blue dashed vertical line describes the largest possible drug {\blue concentration} for the double dose patient (see {\blue\eqref{largeste}}). In both plots, $p=0.8$ and $\alpha=0.9$.}
\label{figsim0}
\end{figure}

Figure~\ref{figsim0}a illustrates how the drug {\blue concentration} in the  body evolves in time. The black dotted curve describes the perfectly adherent patient, and the red dashed curve and blue solid curve describe patients with imperfect adherence following the single dose and double dose protocols, respectively. We set the initial drug concentration equal to $C^{\perf}(0)$ in this illustration.

\subsubsection{Large-time drug {\blue concentration} distribution}

For the case of perfect adherence, the drug concentration at large time is described by $C^{\perf}(t)$ in \eqref{perfper}. Analogously, for the case of imperfect adherence, we prove below that the drug concentration converges in distribution at large time,
\begin{align}\label{Cearly}
c(N\tau+t)
\to_{\dd}
C(t)
\quad\text{as }N\to\infty,
\end{align}
where $t\in[0,\tau)$ is the time since the last dosing time and $C(t)$ is a certain random function given below. In particular, $C(t)$ describes the drug concentration in a patient who has been taking the drug for a long time with adherence $p\in(0,1)$. Furthermore, the patient's drug exposure over a dosing interval is 
\begin{align}\label{auc}
\AUC
:=\int_{0}^{\tau}C(t)\,\dd t.
\end{align}
We emphasize that $C(t)$ and $\AUC$ are random since patient adherence is modeled by a random process.

\subsubsection{The effects of nonadherence}

We measure the effects of nonadherence by comparing the {\blue drug concentration} in a patient with imperfect adherence to the drug {\blue concentration} in a patient with perfect adherence. It is natural to quantify this in terms of the drug exposure ratio $\AUC/\AUC^{\perf}$ or the drug concentration ratio $C(t)/C^{\perf}(t)$ {\blue at some time $t\in[0,\tau)$ since the last scheduled dose}. It turns out that these two ratios are the same, as we prove below that
\begin{align}\label{Z}
Z:=
\frac{\AUC}{\AUC^{\perf}}
=\frac{C(t)}{C^{\perf}(t)}
,\quad\text{for all }t\in[0,\tau).
\end{align}
{\blue We therefore emphasize that we} study the effects of nonadherence {\blue in terms of the relative drug exposure and the relative drug concentration} by studying the single random variable $Z$. Figure~\ref{figsim0}b plots the distribution of $Z$ for the single dose protocol (red) and the double dose protocol (blue) obtained from stochastic simulations (the distribution is obtained from $10^{7}$ realizations of $c(N\tau+t)$ with $N=100$).

We study $Z$ primarily in terms of the following three statistics. First, we define the mean,
\begin{align}\label{mu}
\mu
:=\E[Z]
=\frac{\E[\AUC]}{{\AUC^{\perf}}}
{\blue=\frac{\E[C(t)]}{C^{\perf}(t)},}
\end{align}
which compares the average drug {\blue concentration} to the perfectly adherent patient. We further define the deviation,
\begin{align}\label{deviation0}
\dev
:=\sqrt{\E[(Z-1)^{2}]}
=\frac{\sqrt{\E\big[({\AUC}-{{{\AUC}^{\perf}}})^{2}\big]}}{{{{\AUC}^{\perf}}}}
{\blue=\frac{\sqrt{\E\big[(C(t)-{{C^{\perf}(t)}})^{2}\big]}}{{{C^{\perf}(t)}}}},
\end{align}
which measures how the drug {\blue concentration} deviate{\blue s} from the perfectly adherent patient. In statistics, \eqref{deviation0} is called the relative root-mean-square deviation or relative root-mean-square error. We also compute the coefficient of variation of $Z$, but we find that $\dev$ is a better measure of the effects of nonadherence. Finally, since dosing protocols in which the patient takes more than a single dose at a time may cause drug {\blue concentration}s to rise too high, another useful statistic is the largest possible drug {\blue concentration} compared to the perfectly adherent patient,
\begin{align}\label{lambda0}
\lambda
:=\sup_{\xi}Z
{\blue=\sup_{\xi}\frac{\AUC}{\AUC^{\perf}}
=\sup_{\xi}\frac{C(t)}{C^{\perf}(t)},}
\end{align}
where $\sup_{\xi}$ denotes the supremum over patterns of the patient remembering or forgetting to take their medication (i.e. $\{\xi_{n}\}_{n}$ in \eqref{xin}). {\blue We emphasize that \eqref{lambda0} means that $\lambda$ bounds both the relative drug exposure (i.e.\ $\AUC$ to $\AUC^{\perf}$) and the relative drug concentration at any time (i.e.\ $C(t)$ to $C^{\perf}(t)$).}

We point out that the statistics $\mu$, $\dev$, and $\lambda$ in \eqref{mu}-\eqref{lambda0} are dimensionless, and thus they are independent of the units used to measure drug amounts, concentrations, time, etc. Furthermore, these statistics depend only on $\alpha$ in \eqref{alphabeta}, the adherence percentage $p\in(0,1)$, and the dosing protocol $f_{n}$. {\blue We emphasize that since $\mu$, $\dev$, and $\lambda$ are defined relative to the perfectly adherent patient, their values} are unchanged if $\AUC$ and $\AUC^{\perf}$ are replaced by $C(t)$ and $C^{\perf}(t)$ for any $t\in[0,\tau)$ {\blue(as indicated by \eqref{mu}-\eqref{lambda0})}.

{\blue
\subsubsection{Single and double dose protocols}

In the Appendix, we analyze the mathematical model and compute statistics of the drug {\blue concentration}s for general dosing protocols. Here, we present the formulas for $\mu$, $\dev$, and $\lambda$ for the single and double dose protocols.

\begin{theorem}\label{maintext}
Using superscripts to denote the dosing protocol, the relative means in \eqref{mu} are
\begin{align*}
\mu^{\single}=p,\quad
\mu^{\double}=p+p(1-p).
\end{align*}
Similarly, the deviations in \eqref{deviation0} are
\begin{align}
\dev^{\single}
&=\sqrt{\frac{1-p}{1+\alpha}}\sqrt{1-\alpha  (2 p-1)},\label{es}\\
\dev^{\double}
&=\sqrt{\frac{1-p}{1+\alpha}}\sqrt{1+p+(1-7 p+2 p^2)\alpha+2  p^2(2-p) \alpha ^2}.\label{ed}
\end{align}
Finally, the largest relative drug {\blue concentration}s in \eqref{lambda0} are
\begin{align}\label{largeste}
\lambda^{\single}=1,\quad
\lambda^{\double}=\frac{2}{1+\alpha}.
\end{align}
\end{theorem}

}

\section{Results}

We now explore some pharmacological implications of the analysis above. Recall that $\alpha:=e^{-\ke\tau}$, where $\tau$ is the dosing interval and $\ke$ is the drug elimination rate. Since elimination rates are often expressed in terms of half-lives, we note that the drug half-life, $\thalf>0$, is related to the other parameters via
\begin{align}\label{alphahalf}
\alpha
&=2^{-{\tau}/\thalf},
\quad
\thalf
=\bigg(\frac{\ln\frac{1}{2}}{\ln\alpha}\bigg){\tau}
=\frac{\ln2}{\ke}.
\end{align}
Hence, in the following a ``long drug half-life'' means $\thalf$ is long compared to ${\tau}$, and thus $\alpha$ is large (i.e.\ $\alpha$ is near 1). Similarly, a ``short drug half-life'' means $\thalf$ is short compared to ${\tau}$, and thus $\alpha$ is small.

\subsection{Long half-lives reduce the effects of patient nonadherence}

\begin{figure}[t]
\centering
\includegraphics[width=\linewidth]{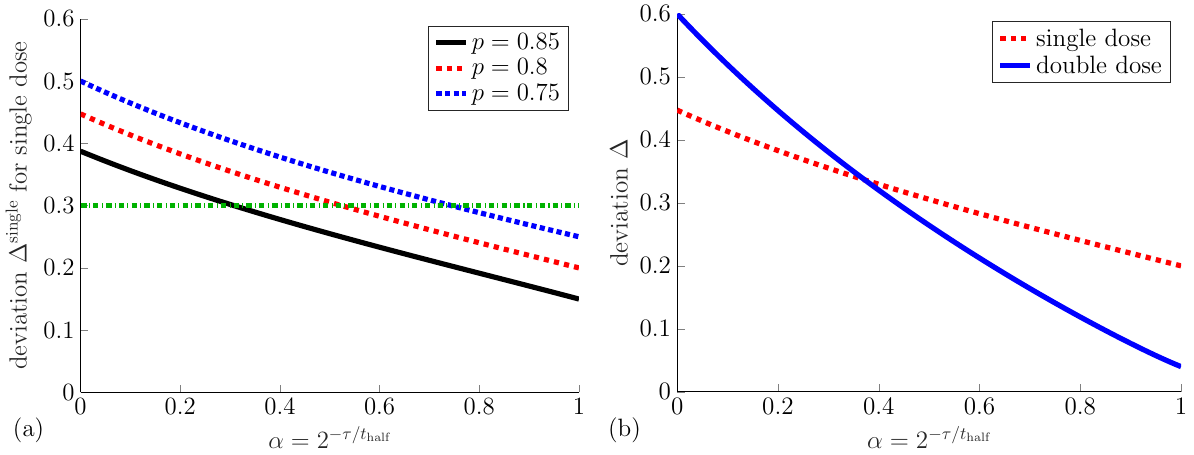}
\caption{Deviation $\dev$ from perfect adherence. Panel (a) plots the deviation $\dev^{\single}$ for the single dose protocol as a function of $\alpha$ for different adherence percentages $p$. Panel (b) compares the deviations $\dev^{\single}$ and $\dev^{\double}$ for the single and double dose protocols with $p=0.8$.
}
\label{fige}
\end{figure}

We begin by considering the single dose protocol. In Figure~\ref{fige}a, we plot $\dev^{\textup{single}}$ as a function of $\alpha$ for different patient adherence levels $p$. As expected, $\dev^{\single}$ decreases as the patient adherence $p$ increases. Furthermore, $\dev^{\single}$ decreases as $\alpha$ increases, and $\dev^{\single}$ approaches its minimum value as $\alpha\to1$,
\begin{align*}
\dev^{\single}\to1-p\quad\text{as }\alpha\to1.
\end{align*}
These properties can be seen from Figure~\ref{fige}a and equation~\eqref{es}.

Importantly, Figure~\ref{fige}a shows that the effect of patient nonadherence, as measured by the deviation $\dev^{\single}$ from perfect adherence, depends critically on $\alpha$. For example, notice that the horizontal line in Figure~\ref{fige}a at $\dev^{\single}=0.25$ intersects the $\dev^{\single}$ curves for the three different levels of patient adherence considered (namely, $p=0.85$, $p=0.8$, and $p=0.75$). Therefore, a patient with high adherence $p$ and small $\alpha$ and a patient with low adherence $p$ and large $\alpha$ can have the same deviation from the perfectly adherent patient.

Put another way, the effects of nonadherence can be lessened by increasing the $\alpha$ value of the drug (i.e.\ increasing the half-life $\thalf$ or decreasing the dosing interval ${\tau}$) without changing the patient's actual adherence $p$. This result is inline with previous analysis, as it is commonly noted that drugs with long half-lives tend to be more ``forgiving'' of missed doses \cite{osterberg2010}. This analysis thus quantifies drug ``forgiveness.'' For other measures of drug forgiveness, see \cite{assawasuwannakit2015, pellock2016, morrison2017}.

\subsection{Double dose protocol mitigates patient nonadherence for drugs with long half-lives}\label{compare}

To compare the deviations from perfect adherence for the single dose and double dose protocols, in Figure~\ref{fige}b we plot $\dev^{\single}$ and $\dev^{\double}$ as functions of $\alpha$. This figure shows that
\begin{align}
\dev^{\single}
&<\dev^{\double}\quad\text{if $\alpha$ is small (i.e.\ short half-life)},\nonumber\\
\text{and}\quad\dev^{\double}
&<\dev^{\single}\quad\text{if $\alpha$ is large (i.e.\ long half-life)}.\label{want}
\end{align}
We set $p=0.8$ in Figure~\ref{fige}b, but other values of $p$ yield similar results. Indeed, the formulas in \eqref{es} and \eqref{ed} imply the small $\alpha$ limits,
\begin{align*}
\lim_{\alpha\to0}\dev^{\single}
=\sqrt{1-p}
< \lim_{\alpha\to0}\dev^{\double}
=\sqrt{1-p^{2}},
\end{align*}
and the large $\alpha$ limits,
\begin{align}\label{ll}
\lim_{\alpha\to1}\dev^{\double}
=(1-p)^{2}
<\lim_{\alpha\to1}\dev^{\single}
=1-p.
\end{align}
In practical terms, \eqref{ll} means that if $\alpha$ is large and the patient has adherence of $p=0.9$, then the deviation from perfect adherence is roughly 10 times smaller for the double dose protocol compared to the single dose protocol.

While we have shown \eqref{want} for large $\alpha$, it follows from \eqref{es}-\eqref{ed} that it is actually the case that
\begin{align}\label{ac}
\dev^{\double}
<\dev^{\single}\quad\text{if and only if }\alpha>\alpha_{\text{c}}
:=\frac{2}{5-2 p+\sqrt{12 (p-3) p+25}}.
\end{align}
It is straightforward to check that the critical value $\alpha_{\text{c}}$ always lies in the interval,
\begin{align*}
\alpha_{\text{c}}\in(0.2,0.5)\quad \text{for all }p\in(0,1).
\end{align*}
Therefore, $\alpha>0.5$ is a sufficient condition for $\dev^{\double}<\dev^{\single}$, and $\alpha>0.5$ is equivalent to $\thalf>{\tau}$.


These results imply that if $\alpha>\alpha_{\text{c}}$, then a patient {\blue following the double dose protocol with actual adherence ${\p}$} can have the same deviation from perfect adherence as they would have by following the single dose protocol with a higher adherence $p_{+}$. {\blue We thus refer to $p_{+}$ as their ``effective adherence.''} To calculate $p_{+}$, suppose the patient has actual adherence ${\p}$. We then find the value of $p_{+}$ which satisfies
\begin{align}\label{set}
\dev^{\single}|_{p_{+}}
=\dev^{\double}|_{{\p}},
\end{align}
where  $\dev^{\single}|_{p_{+}}$ and $\dev^{\double}|_{{\p}}$ denote setting the adherence equal to $p_{+}$ and $\p$ in the respective formulas in \eqref{es}-\eqref{ed}. Solving \eqref{set} yields the ``effective adherence'' $p_{+}\in(0,1)$ as a function of $\alpha$ and the actual adherence ${\p}$. Note that \eqref{ac} implies that $p_{+}>{\p}$ if and only if $\alpha>\alpha_{\text{c}}$.

\begin{figure}[t]
\centering
\includegraphics[width=\linewidth]{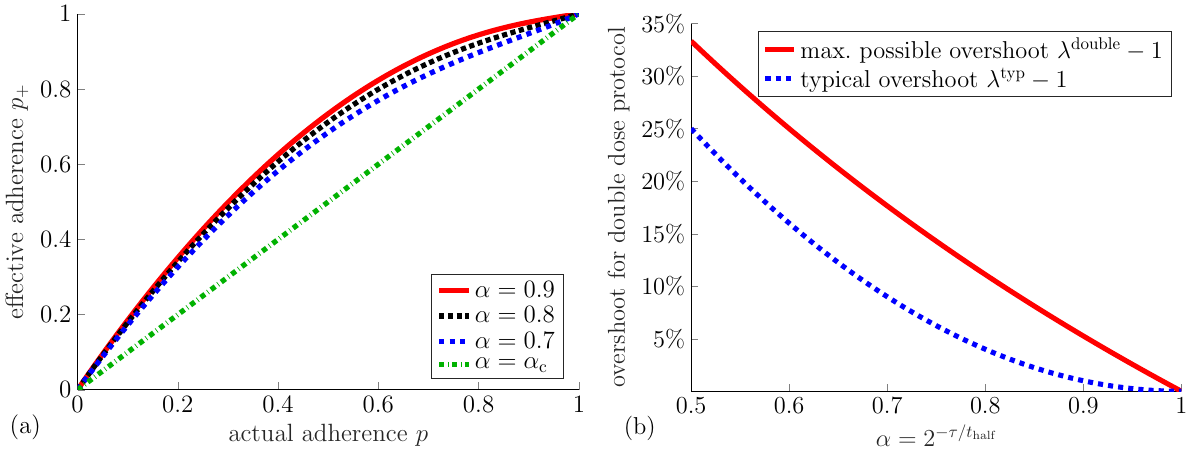}
\caption{Panel (a) plots the effective adherence $p_{+}$ obtained by following the double dose protocol rather than the single dose protocol as a function of the actual adherence $p$. For a patient following the double dose protocol, panel (b) shows how their drug {\blue concentration}s could rise above the {\blue concentration}s in a perfectly adherent patient.}
\label{figp}
\end{figure}

In Figure~\ref{figp}a, we plot $p_{+}$ as a function of ${\p}$ for different values of $\alpha$. This figure shows that the {\blue increase in the effective adherence obtained} by following the double dose protocol is quite substantial, especially if $\alpha$ is close to 1. For example, if $\alpha=0.8$, then a patient with actual adherence of only ${\p}=0.6$ can have effective adherence $p_{+}=0.8$, and a patient with actual adherence of ${\p}=0.8$ can have effective adherence $p_{+}=0.92$. Notice that $p_{+}=\p$ if $\alpha=\alpha_{\textup{c}}$. 

Summarizing, this analysis suggests that (i) the single dose protocol is best when $\thalf\ll {\tau}$ and (ii) the double dose protocol is best when $\thalf\gg {\tau}$. Conclusion (ii) contradicts some common dosing recommendations. Indeed, long drug half-lives are sometimes stated as the reason to avoid the double dose protocol in favor of the single dose protocol (for example, see recommendations for perampanel \cite{albassam2020} and lamotrigine sodium valproate \cite{gilbert2002}). However, we have shown that drugs with long half-lives are precisely the drugs for which patients could benefit from taking a double dose following a missed dose.

\subsection{Double dose protocol is not toxic for drugs with long half-lives}

Taking a double dose is sometimes avoided due to concern that it may cause a toxic drug {\blue concentration} in the body. For a patient following the double dose protocol, {\blue\eqref{largeste}} provides an upper bound to how their drug {\blue concentration or exposure} {\blue($C(t)$ or }$\AUC^{\double}$) could compare to the perfectly adherent patient {\blue($C^{\perf}(t)$ or }$\AUC^{\perf}$). Indeed, {\blue\eqref{largeste}} ensures that
\begin{align}\label{bb}
\begin{split}
\textcolor{black}{\frac{C^{\double}(t)}{C^{\perf}(t)}}
&\textcolor{black}{<\lambda^{\double}
=\frac{2}{1+\alpha},}\quad\text{\textcolor{black}{for any }}\textcolor{black}{t\in[0,\tau)},\\
\frac{\AUC^{\double}}{{\AUC^{\perf}}}
&<\lambda^{\double}
=\frac{2}{1+\alpha}.
\end{split}
\end{align}
{\blue We emphasize that \eqref{bb} means that $\lambda^{\double}$ bounds both the relative drug concentration and the relative drug exposure.}

We plot the maximum possible ``overshoot'' $\lambda^{\double}-1$ in Figure~\ref{figp}b as a function of $\alpha$. Importantly, $\lambda^{\double}$ approaches 1 for large $\alpha$, which means that the possible overshoot from following the double dose protocol vanishes for drugs with long half-lives. In practical terms, \eqref{bb} means that if $\alpha=0.8$, then the drug {\blue concentration} is at most $11\%$ greater than the perfectly adherent patient, and if $\alpha=0.9$, then the drug {\blue concentration} is at most $5\%$ greater than the perfectly adherent patient. 

Furthermore, it is extremely rare for a patient to have a drug {\blue concentration} near the theoretical upper bound in \eqref{bb} if $\alpha$ is large. Indeed, the upper bound in \eqref{bb} is approached only by a patient that alternates exactly between taking and missing the scheduled doses for many dosing intervals if $\alpha$ is large. A more typical overshoot occurs in the following way. If the patient has been taking their medication as prescribed for a long time, then the concentration in their body time $t\in[0,\tau)$ after a dose is roughly the same as the perfectly adherent patient, which is $C^{\perf}(t)$. Then, if they miss one dose and take a double dose at the following dosing time, then the drug concentration time $t\in[0,\tau)$ after the double dose is (compared to $C^{\perf}(t)$) 
\begin{align}\label{typical}
\lambda^{\textup{typ}}
:=\frac{\alpha^{2}C^{\perf}(t)+2\alpha^{t/\tau}\frac{DF}{V}}{C^{\perf}(t)}
=1+(1-\alpha)^{2}
<\lambda^{\double},
\end{align}
which is shown in Figure~\ref{figp}b. 

Summarizing, if $\alpha$ is large, then a patient following the double dose protocol cannot have much more drug in their body than the perfectly adherent patient, where the precise upper bound is in \eqref{bb}. Furthermore, it is rare for the drug {\blue concentration} to approach the upper bound in \eqref{bb} if $\alpha$ is large, and the more typical overshoot is in \eqref{typical}.

\subsection{Adherence thresholds should depend on drug half-life, dosing interval, and dosing protocol}

How much patient adherence is needed for the patient to obtain full treatment benefits? The adherence threshold
\begin{align}\label{standard}
p\ge0.8
\end{align}
has long been considered the definition of an ``adherent patient'' \cite{burnier2019, haynes1976}. However, our calculations show the inadequacy of defining an acceptable patient adherence rate solely in terms of the adherence $p$. To illustrate, suppose patient \#1 has adherence $p_{1}=0.85$ and $\alpha_{1}=0.2$, whereas patient \#2 has adherence $p_{2}=0.75$ and $\alpha_{2}=0.8$. Therefore, by the standard definition in \eqref{standard}, patient \#1 would be deemed ``adherent'' and patient \#2 would be deemed ``non-adherent.'' However, if both patients are following the single dose protocol, then their deviations from the perfectly adherent patient are
\begin{align*}
\dev_{1}^{\single}
&\approx0.33\quad\text{for patient \#1},\\
\text{and}\quad\dev_{2}^{\single}
&\approx0.29\quad\text{for patient \#2}.
\end{align*}
Hence, the supposed ``non-adherent patient'' (patient \#2) is actually closer to the perfectly adherent patient than the ``adherent patient'' (patient \#1).

In fact, the situation is more exasperated if we consider the double dose protocol. To illustrate, suppose patient \#1 and patient \#2 again have respective adherence rates of $p_{1}=0.85 $ and $p_{2}=0.75$, but now suppose they both have $\alpha=0.8$. If patient \#1 follows the single dose protocol and patient \#2 follows the double dose protocol, then their deviations from the perfectly adherent patient are
\begin{align*}
\dev_{1}^{\single}
&\approx0.19\quad\text{for patient \#1},\\
\text{and}\quad\dev_{2}^{\double}
&\approx0.14\quad\text{for patient \#2}.
\end{align*}
Therefore, the drug {\blue concentration}s in the ``non-adherent patient'' are again closer to the perfectly adherent patient than the ``adherent patient.''

\subsection{More complicated dosing protocols}

We now consider more complicated dosing protocols. Notice that in the double dose protocol, the patient never takes more than two doses at a time, even if they missed two or more consecutive prior doses. A more aggressive protocol is the ``triple dose'' protocol in \eqref{ftriple5} in which the patient takes a double dose to make up for a single missed dose and a triple dose to make up for two or more consecutive missed doses. An even more aggressive protocol is the ``all dose'' protocol in \eqref{fall5} in which the patient takes all of their missed doses. As another example, consider the ``fractional'' dosing protocol,
\begin{align}
\begin{split}\label{fractional}
f_{n}^{\textup{frac}}
:=\begin{cases}
0 & \text{if }\xi_{n}=0,\\
1 & \text{if }\xi_{n}=1,\,\xi_{n-1}=1,\\
1+\alpha & \text{if }\xi_{n}=1,\,\xi_{n-1}=0, 
\end{cases}
\end{split}
\end{align}
in which the patient takes an extra large fractional dose if they missed one or more prior doses. The reasoning behind the size of this extra dose is that if the patient had taken their prior dose, then the fraction of that prior dose remaining in their body at the next dosing time would be $\alpha$. Note that \eqref{fractional} is a special case of the boost protocol in \eqref{boost} with $b=\alpha$.

These protocols may be impractical, as \eqref{ftriple5}-\eqref{fall5} require the patient to keep fairly detailed records and \eqref{fractional} requires the ability to take a fractional dose. Nevertheless, it is interesting to consider the implications of these dosing protocols. In the Appendix, we obtain exact analytical formulas for the deviation $\dev$ for these different dosing protocols (see {\blue Corollary~\ref{useful}}). In Figure~\ref{figall}, we plot $\dev$ for these protocols and for the single and double dose protocols as functions of $\alpha$ (we set $p=0.8$). There are two important points to observe from Figure~\ref{figall}.

\begin{figure}[t]
\centering
\includegraphics[width=.6\linewidth]{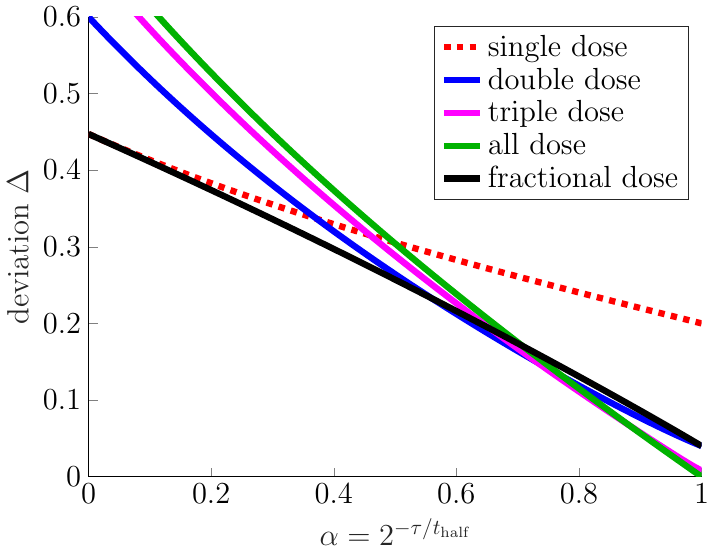}
\caption{Deviation $\dev$ from perfect adherence for various dosing protocols. }
\label{figall}
\end{figure}

First, if $\alpha$ is large, then the deviation $\dev$ is smallest for the triple dose and all dose protocols. Indeed, {\blue the formulas in Corollary~\ref{useful}} imply that
\begin{align*}
\lim_{\alpha\to1}\dev^{\textup{triple}}
&=(1-p)^{3},\quad
\lim_{\alpha\to1}\dev^{\textup{all}}
=0,
\end{align*}
where the superscript indicates the corresponding dosing protocol. Hence, one might recommend the triple dose protocol or even the all dose protocol if $\alpha$ is large. However, while $\dev^{\double}$ is much less than $\dev^{\single}$ for large $\alpha$, the further reductions in the deviation $\dev$ for the triple dose and all dose protocols are comparatively much smaller. Furthermore, compared to the double dose protocol, the triple dose and all dose protocols come with the costs of (i) being more complicated and (ii) allowing for a higher possible drug concentration in the body (see {\blue\eqref{largeste}}).

Next, notice in Figure~\ref{figall} that the fractional dose protocol in \eqref{fractional} results in a deviation $\dev^{\textup{frac}}$ that is near minimal for all values of $\alpha\in(0,1)$. This is perhaps not surprising, since the fractional dose protocol interpolates between the single dose and double dose protocols as $\alpha$ ranges from 0 to 1. Furthermore, it is noteworthy that a patient following the fractional dose protocol is assured to never have too much drug in their body. Indeed, {\blue\eqref{largeste}} implies that the drug exposure $\AUC^{\textup{frac}}$ for a patient following the fractional dose protocol is bounded above by the exposure for the perfectly adherent patient,
\begin{align*}
\AUC^{\textup{frac}}
\le{\AUC^{\perf}}.
\end{align*}
{\blue Similarly, the drug concentration for a patient following the fractional dose protocol is bounded above by the concentration in a perfectly adherent patient,
\begin{align*}
C^{\textup{frac}}(t)
\le C^{\perf}(t),\quad\text{for any }t\in[0,\tau).
\end{align*}
Therefore,} if a patient is able to take fractional doses, then the fractional dose protocol (i) yields a small deviation $\dev$ and (ii) ensures that the patient cannot have more drug in their body than the perfectly adherent patient (regardless of $\alpha$ and $p$).

Of course, the fractional dose protocol is similar to the single dose protocol if $\alpha$ is small, and it is similar to the double dose protocol if $\alpha$ is large. In particular, the fractional dose protocol differs significantly from both the single and double dose protocols only in the case that $\alpha\approx0.5$, which means $\thalf\approx {\tau}$. Therefore, this analysis suggests that (i) the single dose protocol is best when $\thalf\ll {\tau}$, (ii) the double dose protocol is best when $\thalf\gg {\tau}$, and (iii) the ``1.5 dose'' protocol is best when $\thalf\approx {\tau}$, where the 1.5 dose protocol means the patient takes an extra half dose to make up for a missed dose,
\begin{align}
\begin{split}\label{half}
f_{n}^{\textup{half}}
:=\begin{cases}
0 & \text{if }\xi_{n}=0,\\
1 & \text{if }\xi_{n}=1,\,\xi_{n-1}=1,\\
1.5 & \text{if }\xi_{n}=1,\,\xi_{n-1}=0.
\end{cases}
\end{split}
\end{align}
From a practical standpoint, the 1.5 dose protocol may often be feasible to implement (if a standard dose is two pills, then the patient takes three pills if they missed their prior dose). Since \eqref{half} is a special case of \eqref{boost} with $b=0.5$, {\blue\eqref{largeste}} implies that if a patient follows the 1.5 dose protocol, then their drug exposure, $\AUC^{1.5}$, is bounded above by
\begin{align*}
\AUC^{1.5}
\le\begin{cases}
{\AUC^{\perf}} &\text{if }\alpha\ge0.5,\\
\frac{1.5}{1+\alpha}{\AUC^{\perf}} &\text{if }\alpha<0.5.
\end{cases}
\end{align*}
Hence, a patient following the 1.5 dose protocol will never have much more drug in their body than the perfectly adherent patient if $\thalf\approx\tau$.

\subsection{Intuition}\label{intuition}

We have found that the single dose protocol is best when $\thalf\ll {\tau}$ and the double dose protocol is best when $\thalf\gg {\tau}$. These results relied on rather technical mathematical analysis. The purpose of this section is to provide an intuitive explanation for these results.

\subsubsection{Stochastic simulations}

\begin{figure}[t]
\centering
\includegraphics[width=\linewidth]{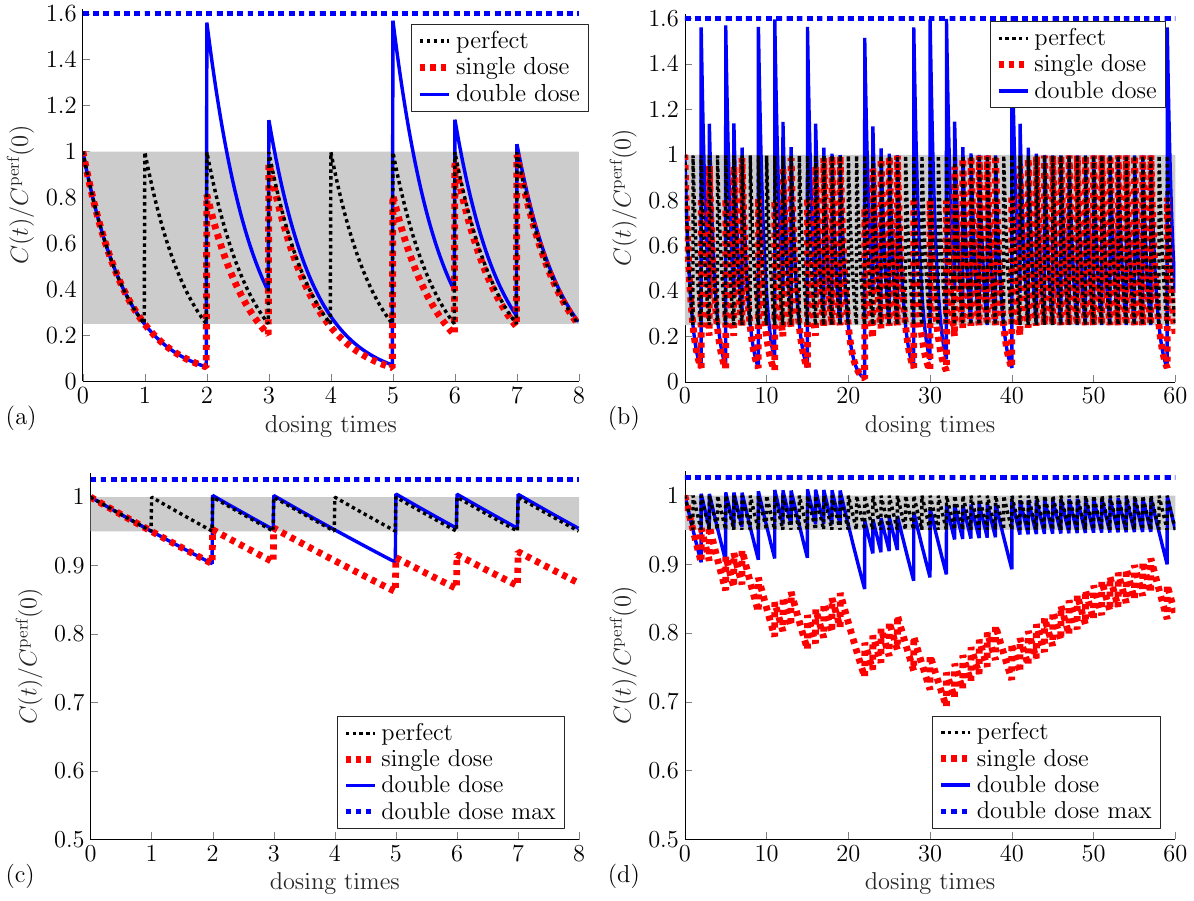}
\caption{Stochastic simulations of drug concentration time courses. In panels (a) and (b), we set $\alpha=0.25$ (meaning $\thalf\ll\tau$), with panel (b) plotted for a long time period. Panels (c) and (d) are the same as (a) and (b), except $\alpha=0.95$ (meaning $\thalf\gg\tau$). The adherence is $p=0.8$ in all panels.}
\label{figsim}
\end{figure}

We begin by plotting stochastic simulations of the drug concentration in the body as a function of time. In Figure~\ref{figsim}a, we set $p=0.8$ and $\alpha=0.25$ (meaning $\thalf\ll\tau$) and plot the concentration under perfect adherence (black dotted curve), and for imperfect adherence for the single dose protocol (red dashed curve) and double dose protocol (blue solid curve). The shaded gray highlights the region between the peaks and troughs for perfect adherence. While this is just one particular realization of the missed doses (the patient happens to miss doses at the first and fourth dosing times), it nevertheless illustrates that the curve for the patient with perfect adherence is better approximated by the single dose protocol than the double dose protocol. Indeed, the single dose and double dose protocols both undershoot the perfect adherence case when a dose is missed, but the double dose protocol then overcompensates when the patient takes their next dose. This is further illustrated in Figure~\ref{figsim}b, which plots the same scenario but for a longer time period. 

\begin{figure}[t]
\centering
\includegraphics[width=\linewidth]{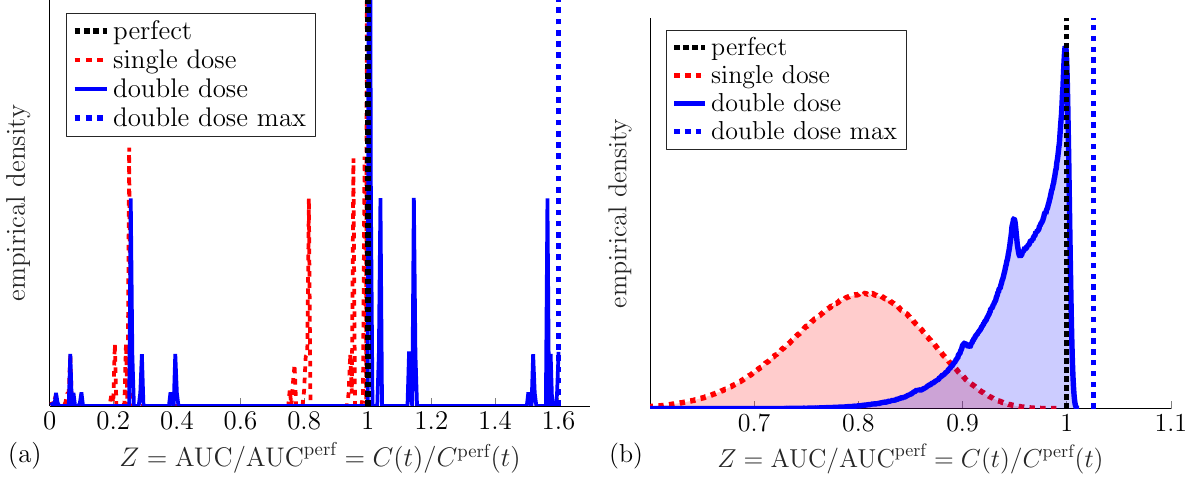}
\caption{The distribution of the drug exposure for the single dose protocol (red) and the double dose protocol (blue) obtained from stochastic simulations. We take $p=0.8$ in both plots and $\alpha=0.25$ in (a) and $\alpha=0.95$ in (b). In both plots, the black dotted vertical line at $\AUC/{\AUC^{\perf}}\textcolor{black}{=C(t)/C^{\perf}(t)}=1$ describes the perfectly adherent patient, and the blue dashed vertical line describes the largest possible drug {\blue concentration} for the double dose patient (see {\blue\eqref{largeste}}).
}
\label{figsim1}
\end{figure}

In Figures~\ref{figsim}c and \ref{figsim}d, we plot the same curves as in Figures~\ref{figsim}a and \ref{figsim}b except in the case that $\alpha=0.95$ (meaning $\thalf\gg\tau$). In this case, the double dose protocol approximates perfect adherence much better than the single dose protocol. While the double dose protocol curve does rise above the perfect adherence curve, it is only by a few percent. In contrast, the single dose protocol curve dips far below the double dose and perfect adherence curves when doses are missed.

In Figure~\ref{figsim1}, we plot the distribution of the drug concentration for the two scenarios in Figure~\ref{figsim}, with $\alpha=0.25$ in Figure~\ref{figsim1}a and $\alpha=0.95$ in Figure~\ref{figsim1}b. The distributions are computed from $10^{7}$ realizations of {\blue$c(N\tau+t)$ with $N=100$}. {\blue The very irregular distributions in Figure~\ref{figsim1}a are typical for small values of $\alpha$. The reason for this irregularity is difficult to intuit and is in fact a rich mathematical topic. Indeed, these irregular distributions for the single dose protocol have been studied in the pure mathematics literature for many decades under the name infinite Bernoulli convolutions \cite{jessen1935, kershner1935, erdos1939, peres2000, solomyak1995, peres1998, escribano2003, hu2008}.}

It is again evident from Figure~\ref{figsim1} that the single dose protocol best approximates the perfectly adherent patient when $\alpha$ is small (i.e.\ short drug half-life), whereas the double dose protocol best approximates the perfectly adherent patient when $\alpha$ is large (i.e.\ long drug half-life). Notice also from Figure~\ref{figsim1}b that it is very rare for the double dose protocol to ever result in a drug {\blue concentration} much larger than the perfectly adherent patient.

\subsubsection{A simple calculation}

The phenomena seen above can be explained with a simple calculation. Suppose the patient has been taking the drug as prescribed for a long time and so the drug concentration time $t\in[0,\tau)$ after a dose is $C^{\perf}(t)$. Suppose the patient then misses one dose and remembers to take the drug at the following dosing time. Under the single dose protocol, the concentration time $t\in[0,\tau)$ after the single dose is
\begin{align*}
\rho^{\single}(t)
:=\alpha^{2}C^{\perf}(t)+\alpha^{t/\tau}\tfrac{DF}{V}
=(1-\alpha(1-\alpha))C^{\perf}(t),
\end{align*}
where we have used that $C^{\perf}(t)=\alpha^{t/\tau}\tfrac{DF}{V} A^{\perf}$ and $A^{\perf}=1/(1-\alpha)$. Alternatively, under the double dose protocol, the concentration time $t\in[0,\tau)$ after the double dose is
\begin{align*}
\rho^{\double}(t)
:=\alpha^{2}C^{\perf}(t)+2\alpha^{t/\tau}\tfrac{DF}{V}
=(1+(1-\alpha)^{2})C^{\perf}(t).
\end{align*}
For small $\alpha$, we have that
\begin{align*}
\rho^{\single}(t)
&\approx(1-\alpha){C^{\perf}(t)},\quad
\rho^{\double}(t)
\approx(2-2\alpha){C^{\perf}(t)},\quad\text{if $\alpha$ is near 0},
\end{align*}
which means the single dose protocol puts the patient slightly below the desired $C^{\perf}(t)$, but the double dose protocol puts the patient at almost twice $C^{\perf}(t)$. However, for large $\alpha$, we have that
\begin{align*}
\rho^{\single}(t)
&\approx(1-(1-\alpha)){C^{\perf}(t)},\\
\rho^{\double}(t)
&=(1+(1-\alpha)^{2}){C^{\perf}(t)},
\quad\text{if $\alpha$ is near 1},
\end{align*}
which means that while the single dose protocol puts the patient below $C^{\perf}(t)$, the double dose protocol puts the patient above $C^{\perf}(t)$ by a much smaller amount. In practical terms, if $\alpha=0.9$, then the single dose patient undershoots $C^{\perf}(t)$ by about $10\%$, whereas the double dose patient overshoots $C^{\perf}(t)$ by a mere $1\%$.

\section{Discussion}

We have formulated and analyzed a mathematical model to investigate how nonadherence to medication affects drug {\blue concentration}s in the body. We computed pharmacologically relevant statistics of the drug {\blue concentration} in the body, thus providing quantitative descriptions of the effects of nonadherence, and how these effects depend on the adherence percentage $p$, drug half-life $\thalf$, dosing interval ${\tau}$, and how missed doses are handled (i.e.\ the dosing protocol). In agreement with previous results \cite{osterberg2005}, we found that drug {\blue concentration}s are less affected by missed doses if the half-life is long compared to the dosing interval, and we quantified this effect. As a general principle, we found that nonadherence is best mitigated by taking double doses following missed doses if the drug half-life is long compared to the dosing interval (i.e.\ $\thalf\gg {\tau}$). Furthermore, in this scenario we found that taking double doses following missed doses cannot cause the drug {\blue concentration} to rise much above the desired {\blue concentration}. Although long drug half-lives are sometimes stated as the reason to avoid a double dose after a missed dose, we have shown that drugs with long half-lives are precisely the drugs for which patients could benefit from taking a double dose after a missed dose.

As an application of these results, consider the synthetic form of thyroxine known as levothyroxine \cite{brent2017}. Levothyroxine is the standard treatment for hypothyroidism, which is one of the most common diseases in the world and affects up to 5\% of the global population \cite{duntas2019}. Levothyroxine pills are used to replace missing thyroid hormone in hypothyroid patients and are usually taken once daily for the remainder of the patient's life \cite{duntas2019}. Hence, the dosing interval is ${\tau}=1\;\text{day}$. The half-life of levothyroxine for hypothyroid patients is between 9 and 10 days \cite{lexicomp, dynamed}, and therefore setting $\thalf=9\;\text{days}$ yields 
\begin{align*}
\alpha=2^{-{\tau}/\thalf}\approx0.93.
\end{align*}
Since this $\alpha$ value is close to 1, our results imply that a hypothyroid patient taking levothyroxine with imperfect adherence can make the drug {\blue concentration} in their body much closer to the {\blue concentration} in a perfectly adherent patient by following the double dose protocol rather than the single dose protocol. That is, if the patient {\blue misses a dose, then it is better to take a double dose at the next dosing time than to skip the missed dose}. These results conflict with common recommendations for levothyroxine, which advise patients to skip any dose that is delayed by more than 12 hours \cite{mayoclinic, rxlist, drugs, nhs}. However, some sources recommend a double dose of levothyroxine after a missed dose (see Chapter 376 of \cite{jameson2018}), and indeed taking a double dose is recognized as safe (see Chapter 36 in \cite{miller2018}). {\blue In fact, the American Thyroid Association has proposed taking up to 7 doses of levothyroxine at once \cite{jonklaas2014}.}

Furthermore, although following the double dose protocol may cause the drug {\blue concentration} in the patient to rise above the {\blue concentration} in a perfectly adherent patient, the maximum possible overshoot for levothyroxine is less than $4\%$ since
\begin{align*}
\lambda^{\double}
=\frac{2}{1+\alpha}
\approx\frac{2}{1+0.93}<1.04.
\end{align*}
In addition, it would be very rare for a patient to have drug {\blue concentration}s near this maximum, as this maximum corresponds to a patient missing doses every other day for many days. Indeed, the typical overshoot is {\blue less than 1\%} for this $\alpha$ value (see \eqref{typical}).

Our model assumes that the drug absorption rate $\ka$ is much faster than the drug elimination rate $\ke$. This is true for most drugs administered orally in conventional dosage forms \cite{gibaldi1982, peletier2017, ma2018, fillastre1987, leroy1990, strandgaarden1999}, including levothyroxine. Indeed, for hypothyroid patients taking levothyroxine, the time to maximum concentration, $t_{\textup{max}}$, is only 3 hours \cite{benvenga1995}, whereas the elimination half-life is $\thalf=9\;\text{days}$ \cite{lexicomp, dynamed}. Using \eqref{alphahalf} and the relation \cite{gibaldi1982},
\begin{align*}
t_{\textup{max}}
=\frac{\ln(\ka/\ke)}{\ka-\ke},
\end{align*}
implies that $\ke/\ka<0.0015$ for levothyroxine.

Several important prior works have used mathematical modeling to investigate the effects of medication nonadherence. Li and Nekka developed stochastic models of the effect of medication nonadherence on patient drug {\blue concentration}s \cite{li2007, li2009}. The models in \cite{li2007, li2009} allow the drug to be administered at irregular times, and the authors obtained analytical formulas for drug {\blue concentration} statistics. In a series of papers \cite{levy2013, fermin2017}, another group of authors developed a variety of stochastic pharmacokinetic models, including ones that allow for variation in dosing times, dose amounts, and elimination rates. The discrete time model proposed in \cite{levy2013} is essentially identical to the model in the present paper in the special case of the single dose protocol. These prior works did not analyze different protocols for handling missed doses. Ma \cite{ma2017} analyzed the mean first passage time for the patient's drug {\blue concentration} to reach a therapeutic range for various ways of handling a missed dose assuming that the patient never misses two or more consecutive doses. Numerical simulations of computational models have also been useful for understanding the effects of nonadherence for specific drugs \cite{gu2020}, especially for antiepileptic drugs \cite{garnett2003, reed2004, dutta2006, ding2012, chen2013, gidal2014, brittain2015, sunkaraneni2018} and antipsychotic drugs \cite{hard2018, elkomy2020}.

Naturally, our model neglects various pharmacological details. We have developed a simple model aimed at addressing patients remembering or forgetting to take their medication{\blue, and we assumed that the patient forgets their medication at each dosing time with a fixed probability, independent of their prior behavior}. However, nonadherence is a dynamic process and patients exhibit a variety of patterns of nonadherence \cite{burnier2019}, including extended ``drug holidays'' \cite{vrijens2008} and ``white-coat adherence'' \cite{burnier2013}. {\blue We also assumed that the patient takes a double dose only if they missed their prior dose. However, actual patients might cause harm by mistakenly taking a double dose when they did not miss their prior dose.} Furthermore, {\blue our model did not allow delayed doses, and a more detailed model would} allow patients to take medication at times that vary continuously. Another source of stochasticity is that pharmacokinetic parameters vary between patients, which has been modeled by analyzing a population of patients with a distribution of parameters \cite{levy2013, fermin2017}.

{\blue Another limitation of our analysis is that we considered only a single compartment pharmacokinetic model with linear elimination and immediate absorption. The pharmacokinetics of some drugs are better described by multicompartment models \cite{gibaldi1982}, and while most drugs can be adequately described by linear processes, there are drugs which exhibit nonlinear kinetics (see Chapter 7 in \cite{gibaldi1982}). In addition, our assumption of fast absorption does not hold for so-called extended release or sustained release drugs \cite{gidal2021, wheless2018, vadivelu2011}. Our model also did not address pharmacodynamics, and an interesting avenue for further research would be to couple the stochastic pharmacokinetic model in this work to a pharmacodynamic model.} 

To conclude, medication nonadherence is a complex and multi-faceted problem, and steps toward its alleviation require contributions from a variety of disciplines. Mathematical modeling is a valuable tool in this endeavor, especially given the ethics of clinical trials that require sporadic dosing. Further, mathematical models can disentangle the effects of various factors and quickly investigate the efficacy of possible interventions. Moving forward, we anticipate that mathematical modeling and analysis will play an important role in understanding and alleviating the effects of medication nonadherence.

\appendix
\section{Appendix}

{\blue
In this appendix, we analyze the mathematical model formulated in the main text.


\subsection{General theory}

Let $\{\xi_{n}\}_{n\in\Z}$ be a bi-infinite sequence of iid Bernoulli random variables as in \eqref{xin} (it is convenient to allow the index $n$ to vary over positive and negative integers). The dose taken at dosing time $n$ may depend on the patient's behavior at time $n$ and the prior $m$ dosing times for some given memory parameter $m\ge0$. Toward this end, let $\{X_{n}\}_{n\in\Z}$ be the history process,
\begin{align}\label{hist0}
X_{n}=(\xi_{n-m},\xi_{n-m+1},\dots,\xi_{n-1},\xi_{n})\in\{0,1\}^{m+1},
\end{align}
which records whether or not the patient remembered at dosing time $n$ and the prior $m$ dosing times. It is immediate that $\{X_{n}\}_{n\in\Z}$ is an irreducible discrete-time Markov chain on the state space $\{0,1\}^{m+1}$ \cite{norris1998}. In particular, let
\begin{align}\label{P}
P
=\{P(x,y)\}_{x,y\in\{0,1\}^{m+1}}
\in\R^{2^{m+1}\times2^{m+1}}
\end{align}
denote the transition probability matrix of the Markov chain $\{X_{n}\}_{n\in\Z}$ with entries defined by
\begin{align*}
P(x,y)
=\P(X_{1}=y\,|\,X_{0}=x),\quad x,y\in\{0,1\}^{m+1},
\end{align*}
where $x\in\{0,1\}^{m+1}$ denotes the vector,
\begin{align*}
x
=(x_{-m},x_{m+1},\dots,x_{-1},x_{0})\in\{0,1\}^{m+1},
\end{align*}
and $y\in\{0,1\}^{m+1}$ is denoted analogously. The definition of $\{\xi_{n}\}_{n\in\Z}$ then implies that the entries of $P$ are
\begin{align*}
P(x,y)
=\begin{cases}
p & \text{if }y_{0}=1,\,(x_{-m+1},\dots,x_{0})=(y_{-m},\dots,y_{-1}),\\
1-p & \text{if }y_{0}=0,\,(x_{-m+1},\dots,x_{0})=(y_{-m},\dots,y_{-1}),\\
0 & \text{otherwise}.
\end{cases}
\end{align*}
Furthermore, the definition of $\{\xi_{n}\}_{n\in\Z}$ implies that the distribution of $X_{n}$ is
\begin{align}\label{pi}
\pi(x)
:=\P(X_{n}=x)
=p^{s(x)}(1-p)^{m+1-s(x)}>0,\quad n\in\Z,\,x\in\{0,1\}^{m+1},
\end{align}
where {\blue$s(x):=\sum_{k=0}^{m}x_{k}\in\{0,1,\dots,m+1\}$} is the number of $1$'s in $x$.

A dosing protocol $f_{n}=f(X_{n})$ is any function
\begin{align}\label{f}
f:\{0,1\}^{m+1}\mapsto[0,\infty).
\end{align}
While we are most interested in the single dose and double dose protocols in \eqref{fsingle} and \eqref{fdouble}, we also investigate a few other protocols. First, consider the ``boost'' dosing protocol,
\begin{align}
\begin{split}\label{boost}
f_{n}^{\textup{boost}}
:=\begin{cases}
0 & \text{if }\xi_{n}=0,\\
1 & \text{if }\xi_{n}=1,\,\xi_{n-1}=1,\\
1+b & \text{if }\xi_{n}=1,\,\xi_{n-1}=0, 
\end{cases}
\end{split}
\end{align}
in which the patient takes a standard single dose of size $D$ plus a ``boost'' dose of size $bD$ if they missed the prior dose, for some $b\ge0$. Notice that the boost protocol reduces to the single dose protocol if $b=0$ and the double dose protocol if $b=1$. Another protocol is the ``triple dose'' protocol,
\begin{align}
\begin{split}\label{ftriple5}
f_{n}^{\textup{triple}}
:=\begin{cases}
0 & \text{if }\xi_{n}=0,\\
1 & \text{if }\xi_{n}=1,\,\xi_{n-1}=1,\\
2 & \text{if }\xi_{n}=1,\,\xi_{n-1}=0,\,\xi_{n-2}=1,\\ 
3 & \text{if }\xi_{n}=1,\,\xi_{n-1}=0,\,\xi_{n-2}=0, 
\end{cases}
\end{split}
\end{align}
in which the patient takes a double dose to make up for a single missed dose and a triple dose to make up for two or more consecutive missed doses. Finally, consider the ``all dose'' protocol in which the patient takes all of their missed doses,
\begin{align}
\begin{split}\label{fall5}
f_{n}^{\textup{all}}
:=\begin{cases}
0 & \text{if }\xi_{n}=0,\\
k+1 & \text{if }\xi_{n}=1,\,\xi_{n-1}=0,\dots,\xi_{n-k}=0,\,\xi_{n-k-1}=1.
\end{cases}
\end{split}
\end{align}
The all dose protocol does not fit into the framework of \eqref{hist0}, and thus an alternative analysis is developed in the section below.

For a dosing protocol $f$, a real number $a\ge0$, integers $M\le N$, and time $t\in[0,\tau)$, define the random variable
\begin{align}\label{CMN}
C_{M,N}(a,t)
:=\alpha^{t/\tau}\frac{DF}{V}\Big(\alpha^{N-M+1}a+\sum_{n=M}^{N}\alpha^{N-n}f(X_{n})\Big),
\end{align}
which is the drug concentration if time $t$ has elapsed since dosing time $n=N$, where $a\ge0$ describes the concentration at dosing time $n=M-1$. We are interested in the drug concentration after a long time, which corresponds to taking $N\to\infty$ in \eqref{CMN}. We will see that this limiting distribution is independent of $a$ and $M$.

Since $\{X_{n}\}_{n\in\Z}$ is a stationary sequence, we have that
\begin{align}\label{shift}
C_{M,N}(a,t)
=_{\dd}C_{-(N-M),0}(a,t),\quad\text{for integers }N\ge M,
\end{align}
where $=_{\dd}$ denotes equality in distribution. Define
\begin{align}\label{Y0}
\begin{split}
{C}(t)
:=\lim_{N\to\infty}C_{-(N-M),0}(a,t)
=\alpha^{t/\tau}\frac{DF}{V}A,\quad\text{for }t\in[0,\tau),
\end{split}
\end{align}
where
\begin{align}\label{AB0}
A
:=\sum_{n=0}^{\infty}\alpha^{n}f(X_{-n}).
\end{align}
The function $f$ must be bounded since the state space $\{0,1\}^{m+1}$ is finite, and thus the Weierstrass M-test ensures that $C(t)$ exists almost surely, and it is immediate that $C(t)$ does not depend on $M\in\Z$ or $a\ge0$. Random variables of the form in \eqref{Y0}-\eqref{AB0} are sometimes called random pullback attractors because they take an initial condition (in this case, $a\ge0$) and pull it back to the infinite past \cite{Crauel01, Mattingly99, Schmalfuss96, lawley15sima, lawley2019hhg}.

Therefore, \eqref{shift} and \eqref{Y0} imply that for any $M\in\Z$ and $a\ge0$, the random variable $C_{M,N}(a,t)$ converges in distribution to $C(t)$ as $N\to\infty$ \cite{billingsley2013}, which we denote by
\begin{align}\label{cd}
C_{M,N}(a,t)\to_{\dd}{C}(t),\quad\text{as }N\to\infty.
\end{align}
Since $f$ is bounded, $C_{M,N}(a,t)$ can be bounded by a nonrandom constant independent of $N$, and thus \eqref{shift}, \eqref{Y0}, and the Lebesgue dominated convergence theorem ensure the convergence of every moment of $C_{M,N}(a,t)$, 
\begin{align}\label{momentconvergence}
\E[(C_{M,N}(a,t))^{j}]
\to\E[({C(t)})^{j}],\quad\text{as }N\to\infty\text{ for all }j>0.
\end{align}
Summarizing, the large $N$ distribution and statistics of $C_{M,N}(a,t)$ are independent of $a\ge0$ and $M\in\Z$, and we can study them by studying the distribution and statistics of $C(t)$.

Furthermore, it is immediate from the definitions in \eqref{perfper} and \eqref{Y0} that
\begin{align}\label{ZZ}
Z
:=\frac{\AUC}{\AUC^{\perf}}
=\frac{C(t)}{C^{\perf}(t)}
=\frac{A}{A^{\perf}},\quad\text{for all }t\in[0,\tau).
\end{align}
Therefore, studying how drug {\blue concentration}s are affected by imperfect adherence amounts to studying $Z$. Since $A^{\perf}=1/(1-\alpha)$, note that
\begin{align*}
\E[Z]=(1-\alpha)\E[A],\quad
\E[Z^{2}]=(1-\alpha)^{2}\E[A^{2}].
\end{align*}

For the single dose protocol in \eqref{fsingle}, the analysis of $Z$ is straightforward since elements of the sequence $\{f(X_{n})\}_{n\in\Z}$ are independent in this special case. The following theorem computes statistics of $Z$ for a general dosing protocol. 

\begin{theorem}\label{thm12}
The first and second moments of $Z$ are
\begin{align}
\E[Z]
&=\sum_{x}f(x)\pi(x),\label{mean}\\
\E[Z^{2}]
&=\frac{1-\alpha}{1+\alpha}\Big(\sum_{x}f(x)2{{u}}(x)-\sum_{x}(f(x))^{2}\pi(x)\Big),\label{secondmoment}
\end{align}
where $\sum_{x}$ denotes the sum over all $x\in\{0,1\}^{m+1}$, $\pi$ is in \eqref{pi}, and
\begin{align}\label{alg2}
{{u}}
:=(I-\alpha P^{\top})^{-1}v,
\end{align}
where $I\in\R^{2^{m+1}\times2^{m+1}}$ is the identity matrix, $P^{\top}$ is the transpose of $P$ in \eqref{P}, and $v\in\R^{2^{m+1}}$ is the vector with entries $v(x)=f(x)\pi(x)$ for $x\in\{0,1\}^{m+1}$.
\end{theorem}

\begin{remark}\label{rmkmath}
The random variable $C(t)$ in \eqref{Y0} generalizes an infinite Bernoulli convolution \cite{peres2000}. If we let $C^{\single}(t)$ denote $C(t)$ in the case that $f$ is the single dose protocol in \eqref{fsingle}, then an infinite Bernoulli convolution is merely a shift and rescaling of $C^{\single}(t)$,
\begin{align*}
\Theta
=\sum_{n=0}^{\infty}\alpha^{n}(2\xi_{n}-1)
=\frac{2C^{\single}(t)-C^{\perf}(t)}{\alpha^{t/\tau}\frac{DF}{V}}.
\end{align*}
Dating back to Erd\H{o}s and others in the 1930s \cite{jessen1935, kershner1935, erdos1939} and continuing in more recent years \cite{peres2000, solomyak1995, peres1998, escribano2003, hu2008}, mathematicians have studied the distribution of $\Theta$. Though the definition of $\Theta$ is quite simple, its distribution is often quite irregular and depends very delicately on the parameters $\alpha$ and $p$.
\end{remark}

\begin{proof}[Proof of Theorem~\ref{thm12}]
Define
\begin{align*}
A_{1}
:=\sum_{n=0}^{\infty}\alpha^{n}f(X_{-n+1}).
\end{align*}
The definition of $A$ in \eqref{AB0} and the stationarity of $\{X_{n}\}_{n\in\Z}$ imply that
\begin{align}\label{inv}
{\A}
=_{\dd}A_{1}
=\alpha {\A}+f(X_{1}),
\end{align}
where $=_{\dd}$ denotes equality in distribution. The invariance relation in \eqref{inv} plays a key role in our analysis.

Taking the expectation of \eqref{inv} and rearranging implies that 
\begin{align}\label{mean0}
\E[{A}]
=\frac{\E[f(X_{1})]}{1-\alpha},
\end{align}
where we have used that ${\A}=_{\dd}A_{1}$. We note that \eqref{mean0} can also be obtained by taking the expectation of \eqref{AB0}. Combining \eqref{pi}, \eqref{ZZ}, and \eqref{mean0} gives \eqref{mean} in Theorem~\ref{thm12}.


Squaring \eqref{inv}, taking expectation, and rearranging implies that
\begin{align}\label{tree0}
\E[A^{2}]
=\frac{1}{1-\alpha^{2}}\Big(2\alpha\E\big[{\A}f(X_{1})\big]+\E\big[(f(X_{1}))^{2}\big]\Big),
\end{align}
where we have again used ${\A}=_{\dd}A_{1}$. By definition of expectation, we have that
\begin{align}\label{tree}
\E\big[(f(X_{1}))^{2}\big]
=\sum_{x}(f(x))^{2}\pi(x).
\end{align}
Computing $\E[{\A}f(X_{1})]$ is more challenging since ${\A}$ and $X_{1}$ are in general correlated if $m\ge1$. 

Let $1_{E}\in\{0,1\}$ denote the indicator function on an event $E$, meaning
\begin{align*}
1_{E}
:=\begin{cases}
1 & \text{if $E$ occurs},\\
0 & \text{otherwise}.
\end{cases}
\end{align*}
Decomposing $\E[{\A}f(X_{1})]$ based on the value of $X_{1}$ gives
\begin{align}\label{tree1}
\E[{\A}f(X_{1})]
=\sum_{x}\E[{\A}f(X_{1})1_{X_{1}=x}]
=\sum_{x}f(x)\E[{\A}1_{X_{1}=x}].
\end{align}
Multiplying \eqref{inv} by the indicator function on the event $X_{1}=x$, taking expectation, and using that $({\A},X_{0})=_{\dd}(A_{1},X_{1})$ yields
\begin{align}\label{c0}
\E[{\A}1_{X_{0}=x}]
=\E[A_{1}1_{X_{1}=x}]
=\alpha\E[{\A}1_{X_{1}=x}]
+f(x)\pi(x),\quad x\in\{0,1\}^{m+1}.
\end{align}

Using the tower property of conditional expectation \cite{durrett2019}, it follows that 
\begin{align}\label{c1}
\E[{\A}1_{X_{1}=x}]
=\sum_{y}\E[{\A}1_{X_{1}=x}1_{X_{0}=y}]
=\sum_{y}\E[{\A}1_{X_{0}=y}]P(y,x),
\end{align}
where $P$ is the transition matrix in \eqref{P}. Combining \eqref{c0} and \eqref{c1} yields the following system of linear algebraic equations for $\E[{\A}1_{X_{0}=x}]$,
\begin{align}\label{c2}
\E[{\A}1_{X_{0}=x}]
=\alpha\sum_{y}\E[{\A}1_{X_{0}=y}]P(y,x)+f(x)\pi(x),\quad x\in\{0,1\}^{m+1}.
\end{align}
If we define the vectors ${{u}},v\in\R^{2^{m+1}}$ by
\begin{align*}
{{u}}(x)
&:=\E[{\A}1_{X_{0}=x}],\quad
v(x)
:=f(x)\pi(x),
\end{align*}
then \eqref{alg2} solves \eqref{c2}. Note that the Perron-Frobenius theorem guarantees that $I-\alpha P^{\top}$ in \eqref{alg2} is invertible since $I-\alpha P^{\top}=\alpha(\alpha^{-1}I-P^{\top})$ and $\alpha\in(0,1)$. Putting this together by combining \eqref{ZZ}, \eqref{alg2} and \eqref{tree0}-\eqref{c0} yields \eqref{secondmoment} in Theorem~\ref{thm12} and completes the proof.
\end{proof}

\subsection{An alternative history process}\label{sectionalt}

The history process in \eqref{hist0} assumes that the patient remembers whether or not they took their medication at the previous $m\ge0$ dosing times. Here, we assume instead that when the patient remembers to take the drug, they know how many consecutive doses they have missed. This modification will allow us to consider the ``all dose'' protocol in \eqref{fall5}.

Define a new history process $\{X_{n}\}_{n\in\Z}$ to encode how much time has passed since the patient last took their medication. Specifically, for integers $n\in\Z$ and $k\ge1$, define
\begin{align}\label{histalt}
X_{n}
&=\begin{cases}
0 & \text{ if }\xi_{n}=\xi_{n-1}=1,\\
k & \text{ if }\xi_{n}=1,\,\xi_{n-1}=\cdots=\xi_{n-k}=0,\,\xi_{n-(k+1)}=1,\\
-k & \text{ if }\xi_{n}=\cdots=\xi_{n-(k-1)}=0,\,\xi_{n-k}=1.
\end{cases}
\end{align}
In words, $X_{n}=0$ if the patient takes the drug at time $n$ and $n-1$, $X_{n}=k\ge1$ if the patient takes the drug at time $n$ after missing the last $k$ doses, and $X_{n}=-k\le-1$ if the patient misses their $k$th consecutive dose at time $n$.

Given that $\{\xi_{n}\}_{n\in\Z}$ are iid as in \eqref{xin}, it follows that $\{X_{n}\}_{n\in\Z}$ is a discrete-time Markov chain on $\mathbb{Z}$ that evolves according to the following transition matrix $P=\{P(x,y)\}_{x,y\in\Z}$ with $P(x,y):=\P(X_{1}=y\,|\,X_{0}=x)$. For $x\in\Z$ and $x\ge0$,
\begin{align}\label{P2}
\begin{split}
P(x,y)
&=\begin{cases}
p & \text{if }y=0,\\
1-p & \text{if }y=-1,\\
0 & \text{otherwise},
\end{cases}
\qquad
P(-x,y)
=\begin{cases}
p & \text{if }y=x,\\
1-p & \text{if }y=-(x+1),\\
0 & \text{otherwise}.
\end{cases}
\end{split}
\end{align}
It is straightforward to check that the distribution of $X_{n}$ is
\begin{align}\label{pi2}
\pi(k)
:=\P(X_{n}=k)
=\begin{cases}
p^{2}(1-p)^{k} & k\ge0,\\
p(1-p)^{|k|} & k\le-1,
\end{cases}
\qquad n,k\in\Z.
\end{align}

For this alternative history process $\{X_{n}\}_{n\in\Z}$ in \eqref{histalt}, a dosing protocol is a function $f:\Z\to[0,\infty)$. Since the state space of $\{X_{n}\}_{n\in\Z}$ (namely $\Z$) is infinite, we assume for technical reasons that dosing protocols cannot grow faster than linearly, \begin{align}\label{bound}
0\le f(k)\le B_{0}|k|+B_{1},\quad k\in\Z,
\end{align}
for some constants $B_{0},B_{1}>0$.

Note that \eqref{bound} and the value of $\pi$ in \eqref{pi2} ensure that all the moments of $X_{n}$ are finite. Furthermore, the bound in \eqref{bound} ensures that the definition of $\A$ in \eqref{Y0} exists almost surely. To see this, note that \eqref{pi2} implies
\begin{align*}
\sum_{n\ge0}\P(|X_{-n}|\ge n)
=\sum_{n\ge0}\sum_{k\ge0}\P(|X_{0}|= n+k)
\le2\sum_{n\ge0}\sum_{k\ge0}p(1-p)^{n+k}
=\frac{2}{p}<\infty.
\end{align*}
Therefore, the Borel-Cantelli lemma \cite{durrett2019} implies that there is an almost surely finite random integer $N_{0}\ge1$ so that $|X_{-n}|<n$ for all $n\ge N_{0}$. Hence, \eqref{bound} implies that $f(X_{-n})\le B_{0}n+B_{1}$ for all $n\ge N_{0}$, and the almost sure existence of $\A$ in \eqref{Y0} follows, as well as the convergence in distribution in \eqref{cd}. Furthermore, the moment convergence in \eqref{momentconvergence} follows from the Lebesgue dominated convergence theorem upon noting that we have almost sure convergence of moments and using some simple bounds on $(C_{M,N}(a,t))^{j}$. 


The definitions of $C_{M,N}(a,t)$, $\A$, and $A_{1}$ and the analysis in the section above carry over directly to this definition of $\{X_{n}\}_{n\in\Z}$ if we use the definition of $P$ and $\pi$ in \eqref{P2} and \eqref{pi2}. In particular, the formula for the mean in \eqref{mean} and the formula for the second moment in \eqref{secondmoment} hold. The benefit of the structure of $P$ in \eqref{P2} is that we can solve for ${{u}}$ in \eqref{alg2} in closed form.

To simplify the formulas for ${{u}}$, we take $f(-k)=0$ for all $k\ge1$, which means the patient cannot take medication when they forget. Equation~\eqref{c2} then implies
\begin{align*}
{{u}}(-k)
&=\alpha(1-p){{u}}(-(k-1)) 
,\quad k\ge2,\\
{{u}}(k)
&=\alpha p {{u}}(-k)+f(k)\pi(k)
,\quad k\ge1,\\
{{u}}(0)
&=\alpha p\sum_{k\ge0}{{u}}(0)+f(0)\pi(0),\\
{{u}}(-1)
&=\alpha(1-p)\sum_{k\ge0}{{u}}(k).
\end{align*}
It is straightforward to solve these equations and obtain that
\begin{align*}
{{u}}(-1)
&=\frac{\alpha  (1-p) (1-\alpha  (1-p)) \sum_{k\ge0}f(k)\pi(k)}{1-\alpha},\\
{{u}}(-k)
&=\alpha^{k-1}(1-p)^{k-1}{{u}}(-1)
,\quad k\ge1,\\
{{u}}(k)
&=\alpha^{k}p(1-p)^{k-1}{{u}}(-1)+f(k)\pi(k),\quad k\ge0.
\end{align*}
Plugging into \eqref{secondmoment} yields
\begin{align}\label{secondmomenta}
\begin{split}
\E[Z^{2}]
&=\frac{2\alpha p(1-\alpha(1-p))}{1-\alpha}\Big(\sum_{k\ge0}f(k)\pi(k)\Big)\sum_{k\ge0}\alpha^{k}(1-p)^{k}f(k)+\sum_{k\ge0}(f(k))^{2}\pi(k).
\end{split}
\end{align}

\subsection{First and second moments of $Z$}\label{examples}

We now work out the first and second moments of $Z$ for a few different choices of the dosing protocol $f$. We begin by considering the history process in \eqref{hist0} with some given memory parameter $m\ge0$. The simplest case is $m=0$, which corresponds to the patient having no recollection of their behavior at prior dosing times. It is natural to suppose that the patient takes no medication when they forget ($f(0)=0$) and that they take their normal dose when they remember ($f(1)=1$). Using \eqref{mean} and \eqref{secondmoment}, we obtain in this case,
\begin{align}
\begin{split}\label{m0}
\E[Z^{\textup{single}}]
&=p,\quad
\E[(Z^{\textup{single}})^{2}]
=\frac{p (1+\alpha  (2 p-1))}{1+\alpha}.
\end{split}
\end{align}

A more interesting case is $m=1$, which allows the patient to potentially take a higher dose if they missed their prior dose. In this case, we need to specify $f(i,j)$ for $i,j\in\{0,1\}$, where $f(i,j)$ is the dose taken at the $n$th dosing time if $\xi_{n}=j$ and $\xi_{n-1}=i$. Let $f(0,0)=f(1,0)=0$ to impose that the patient must miss their dose when they forget. Further, suppose $f(1,1)=1$, which means the patient takes their normal dose if they remember and they did not miss their prior dose. If $f(0,1)=1+b>0$, then we obtain the ``boost'' dosing protocol in \eqref{boost}, and  \eqref{mean} and \eqref{secondmoment} yield
\begin{align}
\begin{split}\label{m1}
\E[Z^{\textup{boost}}]
&=p (1+b (1-p)),\\
\E[(Z^{\textup{boost}})^{2}]
&=\frac{p}{1+\alpha}
\Big[b^2 (p-1) (\alpha +2 \alpha ^2 (p-1) p-1)\\
&\quad+2 b (1-p) (\alpha  (\alpha  p+p-1)+1)+\alpha  (2 p-1)+1\Big].
\end{split}
\end{align}
Note that the cases $b=0$, $b=1$, $b=\alpha$, and $b=0.5$ correspond respectively to the single, double, fractional, and 1.5 dosing protocols in \eqref{fsingle}, \eqref{fdouble}, \eqref{fractional}, and \eqref{half}.

Computing statistics of $Z$ gets more complicated for larger values of $m$. If $m=2$, then we need to specify $f(i,j,k)$ for $i,j,k\in\{0,1\}$, where $f(i,j,k)$ is the dose taken at the $n$th dosing time if $\xi_{n}=k$, $\xi_{n-1}=j$, and $\xi_{n-2}=i$. We set $f(i,j,0)=0$ for $i,j\in\{0,1\}$, and $f(1,1,1)=1$ by the same reasoning as above. It follows then from \eqref{mean} that the mean amount is
\begin{align}\label{m21}
\E[Z]
=p \big[f_{001} (1-p)^2+p (-p (f_{011}+f_{101})+f_{011}+f_{101}+p)\big],
\end{align}
where we have set $f_{ijk}:=f(i,j,k)$ to simplify notation. We can similarly use \eqref{secondmoment} to obtain a complicated, but explicit formula for $\E[Z^{2}]$, which we omit for simplicity. These formulas allow us to investigate dosing protocols in which the patient takes even higher doses following two missed doses compared to a single missed dose. For example, setting
\begin{align}\label{triplevalues}
f(0,1,1)=1,\quad
f(1,0,1)=2,\quad f(0,0,1)=3,
\end{align}
yields the triple dose protocol in \eqref{ftriple5}.


We now consider the alternative history process in \eqref{histalt} in order to consider the ``all dose'' dosing protocol in \eqref{fall5}. In this case, using standard results for summing infinite series, \eqref{mean} and \eqref{secondmomenta} yield
\begin{align}\label{minf}
\begin{split}
\E[Z^{\all}]
&=1,\\
\E[(Z^{\all})^{2}]
&=\frac{\alpha ^2 (2-p) (1-p)+\alpha  (p (p+4)-4)-p+2}{p (1+\alpha) (1-\alpha(1-p))}.
\end{split}
\end{align}

\subsection{Drug {\blue concentration} statistics}

We now use the calculations above to compute pharmacologically relevant statistics. Recall that $\mu=\E[Z]$ in \eqref{mu} compares the mean drug {\blue concentration} to the perfectly adherent patient.

\begin{corollary}\label{cormean}
Using superscripts to denote the dosing protocol, we have that
\begin{align*}
\mu^{\single}
&=p,\quad
\mu^{\double}
=p+p(1-p),\\
\mu^{\textup{boost}}
&=p+bp(1-p),\quad
\mu^{\triple}
=3p-3p^{2}+p^{3},\quad
\mu^{\all}
=1.
\end{align*}
\end{corollary}

\begin{figure}[t]
\centering
\includegraphics[width=\linewidth]{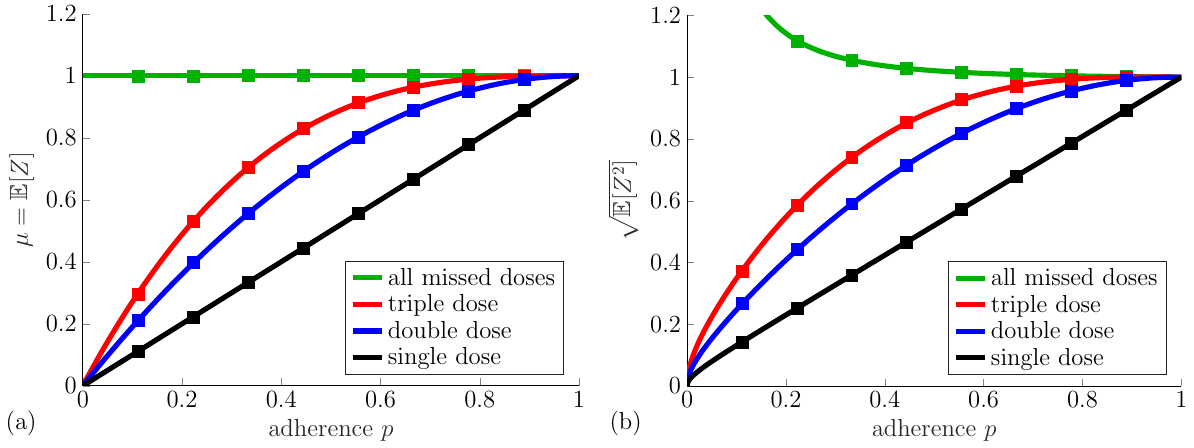}
\caption{The plots compare the mean (panel (a)) and second moment (panel (b)) computed from stochastic simulations (square markers) to the exact analytical formulas as functions of the adherence $p$ for various dosing protocols. We take $\alpha:=e^{-\ke\tau}=0.85$.}
\label{figtest}
\end{figure}

In Figure~\ref{figtest}a, we plot $\mu$ as a function of $p$ for various dosing protocols. Notice that nonadherence causes a reduction in the average drug {\blue concentration}s for the single, double, and triple dose protocols since $\mu^{\single}$, $\mu^{\double}$, and $\mu^{\triple}$ are all strictly less than 1 for all $p\in(0,1)$. Naturally, taking more doses following missed doses increases the average drug {\blue concentration}, and thus
\begin{align*}
\mu^{\single}<\mu^{\double}<\mu^{\triple}<\mu^{\all},\quad \text{for all }p\in(0,1).
\end{align*}
However, notice that while $\mu^{\double}$ is much larger than $\mu^{\single}$, the additional increase is relatively small for the triple and all dose protocols, as $\mu^{\double}\approx\mu^{\triple}\approx\mu^{\all}
$ if the adherence $p$ is not too small.

In Figure~\ref{figtest}b, we plot $\sqrt{\E[Z^{2}]}$ as a function of $p$ for various dosing protocols. The curves in Figure~\ref{figtest} use the analytical formulas for $\E[Z]$ and $\E[Z^{2}]$ obtained from Theorem~\ref{thm12} and the squares markers are results from stochastic simulations with $\alpha=0.85$. In particular, the square markers are obtained from $10^{5}$ independent realizations of $C_{M,N}(a,0)$ with $a=M=0$ and $N=100$. The simulation results agree with the exact analytical results.

To measure the variability in drug {\blue concentration}s that stems from imperfect adherence, we introduce the coefficient of variation of $Z$,
\begin{align}\label{cvauc}
\cv
:=\frac{\sqrt{\E[(Z-\E[Z])^{2}]}}{\E[Z]}
{\blue=\frac{\sqrt{\E[(\AUC-\E[\AUC])^{2}]}}{\E[\AUC]}
=\frac{\sqrt{\E[(C(t)-\E[C(t)])^{2}]}}{\E[C(t)]}},
\end{align}
which is defined as the ratio of the standard deviation to the mean. Notice that \eqref{Z} implies that the coefficient of variation of $Z$ is equal to the coefficient of variation of $\AUC$ or $C(t)$ for any $t\in[0,\tau)$. Applying Theorem~\ref{thm12} gives explicit formulas for the coefficient of variation for the single dose and double dose protocols, which we give in the following corollary (the formulas for the other dosing protocols are omitted for brevity). 

\begin{corollary}\label{cv}
Using superscripts to denote the dosing protocol, we have that
\begin{align*}
\cv^{\single}
&=\sqrt{\frac{1-\alpha}{1+\alpha}}\sqrt{p(1-p)},\\
\cv^{\double}
&=\cv^{\single}\sqrt{4-3 p+p^2-2 p  (2-p) \alpha}.
\end{align*}
\end{corollary}

The coefficient of variation measures the variability induced by nonadherence by measuring how drug {\blue concentration}s deviate from their average value. From a pharmacological standpoint, a small coefficient of variation is desirable. However, a small coefficient of variation does not necessarily imply that the effects of nonadherence are small. Indeed, the coefficient of variation vanishes if the patient never takes their medication ($p=0$).

Hence, a more useful statistic for measuring the effects of nonadherence is how drug {\blue concentration}s deviate from the drug {\blue concentration}s of a perfectly adherent patient, which is the deviation $\dev$ defined in \eqref{deviation0}. Applying Theorem~\ref{thm12} gives explicit formulas for the deviation $\dev$ for different dosing protocols, which we give in the following corollary. 

\begin{corollary}\label{useful}
The deviation in \eqref{deviation0} for the boost protocol in \eqref{boost} for any $b\ge0$ is
\begin{align*}
\dev^{\textup{boost}}
&=\sqrt{\frac{1-p}{1+\alpha}}\sqrt{\alpha +p \left(b^2+2 \alpha ^2 b p (1+b-
bp)-\alpha  (b (b-2 p+4)+2)\right)+1}.
\end{align*}
Note that setting $b=0$, $b=1$, $b=\alpha$, and $b=0.5$ yield the respective deviations for the single, double, fractional, and 1.5 dosing protocols in \eqref{fsingle}, \eqref{fdouble}, \eqref{fractional}, and \eqref{half}. The deviation for the triple dose protocol in \eqref{ftriple5} is
\begin{align}
\dev^{\triple}
&=\sqrt{\frac{1-p}{1+\alpha}}\Big[1-3 p^2+4 p+2 \alpha ^3 (1-p) p^2 ((p-3) p+3)\\
&\quad+2 \alpha ^2 p^2 ((p-3) p+3)+\alpha  \left(1-2 p^3+11 p^2-14 p\right)\Big]^{1/2}.
\end{align}
The deviation for the all dose protocol in \eqref{fall5} is
\begin{align*}
\dev^{\all}
&=\sqrt{\frac{1-p}{1+\alpha}}\sqrt{\frac{2(1-\alpha)^{2}}{1-p (\alpha  (1-p))}}.
\end{align*}
\end{corollary}

For certain drugs, it is important to ensure that the dosing protocol cannot cause the drug {\blue concentration} in the patient to rise too high. We thus consider $\lambda$ in \eqref{lambda0}, which is the largest possible drug {\blue concentration} compared to the perfectly adherent patient. The following theorem calculates $\lambda$ for the dosing protocols above.

\begin{theorem}\label{largest}
Using superscripts to denote the dosing protocol, we have that
\begin{align*}
\lambda^{\textup{boost}}
&=\max\Big\{1,\frac{1+b}{1+\alpha}\Big\},\quad
\lambda^{\triple}
=\frac{3}{1+\alpha+\alpha^{2}},\quad
\lambda^{\all}
=\infty.
\end{align*}
Note that setting $b=0$, $b=1$, $b=\alpha$, and $b=0.5$ corresponds respectively to the single, double, fractional, and 1.5 dosing protocols in \eqref{fsingle}, \eqref{fdouble}, \eqref{fractional}, and \eqref{half}.
\end{theorem}

Notice that if we set $b=\alpha$ in the boost protocol in \eqref{boost}, then $\lambda^{\textup{boost}}=1$ and thus Theorem~\ref{largest} ensures that a patient following the boost protocol with $b=\alpha$ will never have more drug in their body than the perfectly adherent patient.  

We note that Theorem~\ref{maintext} in the main text follows immediately from Corollaries~\ref{cormean} and \ref{useful} and Theorem~\ref{largest}.


\begin{proof}[Proof of Theorem~\ref{largest}]
For the single dose protocol, it is immediate that $\lambda^{\single}=1$, and this corresponds to a patient who never misses a dose.

For the double dose protocol, observe that if $\xi_{2n}=1$ and $\xi_{2n+1}=0$ for all $n\in\Z$, then $A=\frac{2}{1-\alpha^{2}}$, and thus
\begin{align}\label{lb}
\lambda^{\double}
\ge\frac{2/(1-\alpha^{2})}{{A^{\perf}}}
=\frac{2}{1+\alpha}.
\end{align}
This describes a patient who misses a dose at every odd dosing time, and thus always takes a double dose at even dosing times.

To see that $\lambda^{\double}\le\frac{2}{1+\alpha}$, suppose that the patient has concentration $\frac{2}{1-\alpha^{2}}\frac{DF}{V}$ just after dosing time $n=1$. If they take the drug at dosing time $n=2$, then the concentration in their body will be lower than $\frac{2}{1-\alpha^{2}}\frac{DF}{V}$ since
\begin{align*}
\alpha\frac{2}{1-\alpha^{2}}+1<\frac{2}{1-\alpha^{2}}.
\end{align*}
Hence, suppose they miss taking the drug at dosing time $n=2$. If they take the drug at dosing time $n=3$, then they will take a double dose and the concentration in their body will return to $\frac{2}{1-\alpha^{2}}\frac{DF}{V}$. If they miss taking the drug at dosing time $n=3$, then the concentration will be even lower. Therefore, $\lambda^{\double}\le\frac{2}{1+\alpha}$, which upon combining with \eqref{lb} yields $\lambda^{\double}=\frac{2}{1+\alpha}$.

The proof that $\lambda^{\textup{boost}}=\max\{1,\frac{1+b}{1+\alpha}\}$ is almost identical to the proof that $\lambda^{\double}=\frac{2}{1+\alpha}$. The proof that $\lambda^{\triple}=\frac{3}{1+\alpha+\alpha^{2}}$ is also almost identical, upon noting that this value of $\lambda^{\triple}$ is attained by a patient who takes medication at every third dosing time.

The proof for the ``all dose'' protocol follows from noting that if the patient misses $k$ consecutive doses and then takes the next dose, then the drug concentration in their body just after that dose is at least $(k+1)\frac{DF}{V}$. Since this is true for every positive integer $k$, the result $\lambda^{\all}=\infty$ follows.
\end{proof}

}

\subsubsection*{Acknowledgments}
SDL was supported by the National Science Foundation (Grant Nos.\ DMS-1944574 and DMS-1814832). SDL thanks Jennifer Babin and Colt Schisler for helpful discussions.


\bibliography{library.bib}

\begin{thebibliography}{10}

\bibitem{vrijens2012}
Bernard Vrijens, Sabina De~Geest, Dyfrig~A Hughes, Kardas Przemyslaw, Jenny
  Demonceau, Todd Ruppar, Fabienne Dobbels, Emily Fargher, Valerie Morrison,
  Pawel Lewek, et~al.
\newblock A new taxonomy for describing and defining adherence to medications.
\newblock {\em British journal of clinical pharmacology}, 73(5):691--705, 2012.

\bibitem{osterberg2005}
Lars Osterberg and Terrence Blaschke.
\newblock Adherence to medication.
\newblock {\em New England journal of medicine}, 353(5):487--497, 2005.

\bibitem{who2003}
Eduardo Sabat{\'e}, Eduardo Sabat{\'e}, et~al.
\newblock {\em Adherence to long-term therapies: evidence for action}.
\newblock World Health Organization, 2003.

\bibitem{haynes2002}
R~Brian Haynes, Heather~Pauline McDonald, Amit Garg, and Patty Montague.
\newblock Interventions for helping patients to follow prescriptions for
  medications.
\newblock {\em Cochrane database of systematic reviews}, (2), 2002.

\bibitem{burnier2019}
Michel Burnier.
\newblock Is there a threshold for medication adherence? lessons learnt from
  electronic monitoring of drug adherence.
\newblock {\em Frontiers in pharmacology}, 9:1540, 2019.

\bibitem{lindenfeld2017}
Joann Lindenfeld and Mariell Jessup.
\newblock `{Drugs don't work in patients who don't take them' (C. Everett Koop,
  MD, US Surgeon General, 1985)}.
\newblock {\em European journal of heart failure}, 19(11):1412--1413, 2017.

\bibitem{geest2018}
Sabina De~Geest, Leah~L Zullig, Jacqueline Dunbar-Jacob, Remon Helmy, Dyfrig~A
  Hughes, Ira~B Wilson, and Bernard Vrijens.
\newblock {ESPACOMP medication adherence reporting guideline (EMERGE)}.
\newblock {\em Annals of internal medicine}, 169(1):30--35, 2018.

\bibitem{spilker1991}
Bert Spilker and JA~Cramer.
\newblock Patient compliance in medical practice and clinical trials.
\newblock {\em New York}, 1991.

\bibitem{barfod2006}
TS~Barfod, Henrik~Toft S{\o}rensen, Henrik Nielsen, Lotte Rodkj{\ae}r, and
  Niels Obel.
\newblock {`Simply forgot' is the most frequently stated reason for missed
  doses of HAART irrespective of degree of adherence}.
\newblock {\em HIV medicine}, 7(5):285--290, 2006.

\bibitem{howard1999}
J~Howard, K~Wildman, J~Blain, S~Wills, and D~Brown.
\newblock The importance of drug information from a patient perspective.
\newblock {\em Journal of Social and Administrative Pharmacy},
  16(3/4):115--126, 1999.

\bibitem{albassam2020}
Abdullah Albassam and Dyfrig~A Hughes.
\newblock What should patients do if they miss a dose? {A} systematic review of
  patient information leaflets and summaries of product characteristics.
\newblock {\em European Journal of Clinical Pharmacology}, pages 1--10, 2020.

\bibitem{gilbert2002}
Andrew Gilbert, Libby Roughead, Lloyd Sansom, et~al.
\newblock {I've missed a dose; what should I do?}
\newblock {\em Australian Prescriber}, 25(1):16--17, 2002.

\bibitem{peres2000}
Yuval Peres, Wilhelm Schlag, and Boris Solomyak.
\newblock Sixty years of {B}ernoulli convolutions.
\newblock In {\em Fractal geometry and stochastics II}, pages 39--65. Springer,
  2000.

\bibitem{solomyak1995}
Boris Solomyak.
\newblock On the random series $\sum\pm\lambda^{n}$ (an {E}rdos problem).
\newblock {\em Annals of Mathematics}, pages 611--625, 1995.

\bibitem{peres1998}
Yuval Peres and Boris Solomyak.
\newblock Self-similar measures and intersections of {C}antor sets.
\newblock {\em Transactions of the American Mathematical Society},
  350(10):4065--4087, 1998.

\bibitem{escribano2003}
C~Escribano, MA~Sastre, and E~Torrano.
\newblock Moments of infinite convolutions of symmetric bernoulli
  distributions.
\newblock {\em Journal of computational and applied mathematics},
  153(1-2):191--199, 2003.

\bibitem{hu2008}
Tian-You Hu and Ka-Sing Lau.
\newblock Spectral property of the {B}ernoulli convolutions.
\newblock {\em Advances in Mathematics}, 219(2):554--567, 2008.

\bibitem{jessen1935}
B{\o}rge Jessen and Aurel Wintner.
\newblock Distribution functions and the riemann zeta function.
\newblock {\em Transactions of the American Mathematical Society},
  38(1):48--88, 1935.

\bibitem{kershner1935}
Richard Kershner and Aurel Wintner.
\newblock On symmetric bernoulli convolutions.
\newblock {\em American Journal of Mathematics}, 57(3):541--548, 1935.

\bibitem{erdos1939}
Paul Erd{\"o}s.
\newblock On a family of symmetric {B}ernoulli convolutions.
\newblock {\em American Journal of Mathematics}, 61(4):974--976, 1939.

\bibitem{levy2013}
Pierre-Emmanuel L{\'e}vy-V{\'e}hel and Jacques L{\'e}vy-V{\'e}hel.
\newblock Variability and singularity arising from poor compliance in a
  pharmacokinetic model i: the multi-iv case.
\newblock {\em Journal of pharmacokinetics and pharmacodynamics}, 40(1):15--39,
  2013.

\bibitem{fermin2017}
Lisandro~J Ferm{\'\i}n and Jacques L{\'e}vy-V{\'e}hel.
\newblock Variability and singularity arising from poor compliance in a
  pharmacokinetic model ii: the multi-oral case.
\newblock {\em Journal of mathematical biology}, 74(4):809--841, 2017.

\bibitem{Crauel01}
H.~Crauel.
\newblock Random point attractors versus random set attractors.
\newblock {\em J. London Math. Soc. (2)}, 63(2):413--427, 2001.

\bibitem{Mattingly99}
J.~C. Mattingly.
\newblock Ergodicity of {$2$}{D} {N}avier-{S}tokes equations with random
  forcing and large viscosity.
\newblock {\em Comm Math Phys}, 206(2):273--288, 1999.

\bibitem{Schmalfuss96}
B.~Schmalfu{\ss}.
\newblock A random fixed point theorem based on {L}yapunov exponents.
\newblock {\em Random Comput. Dynam.}, 4(4):257--268, 1996.

\bibitem{lawley15sima}
S.~D. Lawley, J.~C. Mattingly, and M.~C. Reed.
\newblock Stochastic switching in infinite dimensions with applications to
  random parabolic {PDE}.
\newblock {\em SIAM J Math Anal}, 47(4):3035--3063, 2015.

\bibitem{lawley2019hhg}
Sean~D Lawley and James~P Keener.
\newblock Electrodiffusive flux through a stochastically gated ion channel.
\newblock {\em SIAM Journal on Applied Mathematics}, 79(2):551--571, 2019.

\bibitem{siwale2016}
Rodney~C. Siwale and Shabnam~N. Sani.
\newblock {\em Multiple-Dosage Regimens}.
\newblock McGraw-Hill Education, New York, NY, 2016.

\bibitem{gibaldi1982}
M~Gibaldi and D~Perrier.
\newblock {\em Pharmacokinetics}.
\newblock Marcelly Dekker, 2 edition, 1982.

\bibitem{bauer2015}
Larry~A. Bauer.
\newblock {\em Clinical Pharmacokinetic Equations and Calculations}.
\newblock McGraw-Hill Medical, New York, NY, 2015.

\bibitem{peletier2017}
Lambertus~A Peletier and Willem de~Winter.
\newblock Impact of saturable distribution in compartmental pk models: dynamics
  and practical use.
\newblock {\em Journal of pharmacokinetics and pharmacodynamics}, 44(1):1--16,
  2017.

\bibitem{ma2018}
Jie Ma and James~P Keener.
\newblock The computation of biomarkers in pharmacokinetics with the aid of
  singular perturbation methods.
\newblock {\em Journal of mathematical biology}, 77(5):1407--1430, 2018.

\bibitem{fillastre1987}
Jean~Paul Fillastre, Annie Leroy, and G~Humbert.
\newblock Ofloxacin pharmacokinetics in renal failure.
\newblock {\em Antimicrobial agents and chemotherapy}, 31(2):156--160, 1987.

\bibitem{leroy1990}
A~Leroy, JP~Fillastre, and G~Humbert.
\newblock Lomefloxacin pharmacokinetics in subjects with normal and impaired
  renal function.
\newblock {\em Antimicrobial agents and chemotherapy}, 34(1):17--20, 1990.

\bibitem{strandgaarden1999}
Kerstin Strandg{\aa}rden, Peter H{\"o}glund, {\"O}rjan Nordle, Jiri Polacek,
  Hans W{\"a}nnman, and Per-Olov Gunnarsson.
\newblock Dissolution rate-limited absorption and complete bioavailability of
  roquinimex in man.
\newblock {\em Biopharmaceutics \& drug disposition}, 20(7):347--354, 1999.

\bibitem{osterberg2010}
LG~Osterberg, J~Urquhart, and TF~Blaschke.
\newblock Understanding forgiveness: minding and mining the gaps between
  pharmacokinetics and therapeutics.
\newblock {\em Clinical Pharmacology \& Therapeutics}, 88(4):457--459, 2010.

\bibitem{assawasuwannakit2015}
P~Assawasuwannakit, R~Braund, and SB~Duffull.
\newblock Quantification of the forgiveness of drugs to imperfect adherence.
\newblock {\em CPT: pharmacometrics \& systems pharmacology}, 4(3):204--211,
  2015.

\bibitem{pellock2016}
John~M Pellock and Scott~T Brittain.
\newblock Use of computer simulations to test the concept of dose forgiveness
  in the era of extended-release {(XR)} drugs.
\newblock {\em Epilepsy \& Behavior}, 55:21--23, 2016.

\bibitem{morrison2017}
Alan Morrison, Melissa~E Stauffer, and Anna~S Kaufman.
\newblock Relationship between adherence rate threshold and drug `forgiveness'.
\newblock {\em Clinical Pharmacokinetics}, 56(12):1435--1440, 2017.

\bibitem{haynes1976}
R~Brian Haynes.
\newblock A critical review of ``determinants'' of patient compliance with
  therapeutic regimens.
\newblock {\em Compliance with therapeutic regimens}, pages 26--39, 1976.

\bibitem{brent2017}
Gregory~A. Brent and Ronald~J. Koenig.
\newblock {\em Thyroid and Antithyroid Drugs}.
\newblock McGraw-Hill Education, New York, NY, 2017.

\bibitem{duntas2019}
Leonidas~H Duntas and Jacqueline Jonklaas.
\newblock Levothyroxine dose adjustment to optimise therapy throughout a
  patient?s lifetime.
\newblock {\em Advances in Therapy}, pages 1--17, 2019.

\bibitem{lexicomp}
Lexicomp.
\newblock Levothyroxine.
\newblock Data retrieved on February 4, 2021 from {http://online.lexi.com}.

\bibitem{dynamed}
{D}yna{M}ed.
\newblock Levothyroxine.
\newblock Data retrieved on February 4, 2021 from {http://dynamed.com}.

\bibitem{mayoclinic}
Mayo {C}linic.
\newblock Levothyroxine.
\newblock Data retrieved on February 4, 2021 from
  {https://www.mayoclinic.org/drugs-supplements/levothyroxine-oral-route/proper-use/drg-20072133}.

\bibitem{rxlist}
{R}x{L}ist.
\newblock Synthroid.
\newblock Data retrieved on February 4, 2021 from
  {https://www.rxlist.com/synthroid-drug/patient-images-side-effects.htm}.

\bibitem{drugs}
Drugs.com.
\newblock Levothyroxine.
\newblock Data retrieved on February 4, 2021 from
  {https://www.drugs.com/levothyroxine.html}.

\bibitem{nhs}
National Health~Service (UK).
\newblock Levothyroxine.
\newblock Data retrieved on February 4, 2021 from
  {https://www.nhs.uk/medicines/levothyroxine/}.

\bibitem{jameson2018}
J.~Larry Jameson, Susan~J. Mandel, and Anthony~P. Weetman.
\newblock {\em Hypothyroidism}.
\newblock McGraw-Hill Education, New York, NY, 2018.

\bibitem{miller2018}
Shannon~A. Miller.
\newblock {\em Thyroid Disorders MTM Data Set}.
\newblock McGraw-Hill Education, New York, NY, 2018.

\bibitem{jonklaas2014}
Jacqueline Jonklaas, Antonio~C Bianco, Andrew~J Bauer, Kenneth~D Burman, Anne~R
  Cappola, Francesco~S Celi, David~S Cooper, Brian~W Kim, Robin~P Peeters,
  M~Sara Rosenthal, et~al.
\newblock {Guidelines for the treatment of hypothyroidism: prepared by the
  American Thyroid Association task force on thyroid hormone replacement}.
\newblock {\em Thyroid}, 24(12):1670--1751, 2014.

\bibitem{benvenga1995}
Salvatore Benvenga, Luigi Bartolone, Stefano Squadrito, Francesco~Lo Giudice,
  and Francesco Trimarchi.
\newblock Delayed intestinal absorption of levothyroxine.
\newblock {\em Thyroid}, 5(4):249--253, 1995.

\bibitem{li2007}
Jun Li and Fahima Nekka.
\newblock A pharmacokinetic formalism explicitly integrating the patient drug
  compliance.
\newblock {\em Journal of pharmacokinetics and pharmacodynamics},
  34(1):115--139, 2007.

\bibitem{li2009}
Jun Li and Fahima Nekka.
\newblock A probabilistic approach for the evaluation of pharmacological effect
  induced by patient irregular drug intake.
\newblock {\em Journal of pharmacokinetics and pharmacodynamics},
  36(3):221--238, 2009.

\bibitem{ma2017}
Jie Ma.
\newblock {\em Stochastic Modeling of Random Drug Taking Processes and the Use
  of Singular Perturbation Methods in Pharmacokinetics}.
\newblock PhD thesis, The University of Utah, 2017.

\bibitem{gu2020}
Jia-qin Gu, Yun-peng Guo, Zheng Jiao, Jun-jie Ding, and Guo-Fu Li.
\newblock How to handle delayed or missed doses: a population pharmacokinetic
  perspective.
\newblock {\em European journal of drug metabolism and pharmacokinetics},
  45(2):163--172, 2020.

\bibitem{garnett2003}
William~R Garnett, Angus~M McLean, Yuxin Zhang, Susan Clausen, and Simon~J
  Tulloch.
\newblock Simulation of the effect of patient nonadherence on plasma
  concentrations of carbamazepine from twice-daily extended-release capsules.
\newblock {\em Current medical research and opinion}, 19(6):519--525, 2003.

\bibitem{reed2004}
Ronald~C Reed and Sandeep Dutta.
\newblock Predicted serum valproic acid concentrations in patients missing and
  replacing a dose of extended-release divalproex sodium.
\newblock {\em American journal of health-system pharmacy}, 61(21):2284--2289,
  2004.

\bibitem{dutta2006}
S~Dutta and RC~Reed.
\newblock Effect of delayed and/or missed enteric-coated divalproex doses on
  valproic acid concentrations: simulation and dose replacement recommendations
  for the clinician 1.
\newblock {\em Journal of clinical pharmacy and therapeutics}, 31(4):321--329,
  2006.

\bibitem{ding2012}
Jun-jie Ding, Yun-jian Zhang, Zheng Jiao, and Yi~Wang.
\newblock The effect of poor compliance on the pharmacokinetics of
  carbamazepine and its epoxide metabolite using monte carlo simulation.
\newblock {\em Acta Pharmacologica Sinica}, 33(11):1431--1440, 2012.

\bibitem{chen2013}
Chao Chen, James Wright, Barry Gidal, and John Messenheimer.
\newblock Assessing impact of real-world dosing irregularities with lamotrigine
  extended-release and immediate-release formulations by pharmacokinetic
  simulation.
\newblock {\em Therapeutic drug monitoring}, 35(2):188--193, 2013.

\bibitem{gidal2014}
Barry~E Gidal, Oneeb Majid, Jim Ferry, Ziad Hussein, Haichen Yang, Jin Zhu,
  Randi Fain, and Antonio Laurenza.
\newblock The practical impact of altered dosing on perampanel plasma
  concentrations: pharmacokinetic modeling from clinical studies.
\newblock {\em Epilepsy \& Behavior}, 35:6--12, 2014.

\bibitem{brittain2015}
Scott~T Brittain and James~W Wheless.
\newblock Pharmacokinetic simulations of topiramate plasma concentrations
  following dosing irregularities with extended-release vs. immediate-release
  formulations.
\newblock {\em Epilepsy \& Behavior}, 52:31--36, 2015.

\bibitem{sunkaraneni2018}
Soujanya Sunkaraneni, David Blum, Elizabeth Ludwig, Vaishali Chudasama, Jill
  Fiedler-Kelly, Marketa Marvanova, Jacquelyn Bainbridge, and Luann Phillips.
\newblock Population pharmacokinetic evaluation and missed-dose simulations for
  eslicarbazepine acetate monotherapy in patients with partial-onset seizures.
\newblock {\em Clinical pharmacology in drug development}, 7(3):287--297, 2018.

\bibitem{hard2018}
Marjie~L Hard, Angela~Y Wehr, Brian~M Sadler, Richard~J Mills, and Lisa von
  Moltke.
\newblock Population pharmacokinetic analysis and model-based simulations of
  aripiprazole for a 1-day initiation regimen for the long-acting antipsychotic
  aripiprazole lauroxil.
\newblock {\em European journal of drug metabolism and pharmacokinetics},
  43(4):461--469, 2018.

\bibitem{elkomy2020}
Mohammed~H Elkomy.
\newblock Changing the drug delivery system: Does it add to non-compliance
  ramifications control? a simulation study on the pharmacokinetics and
  pharmacodynamics of atypical antipsychotic drug.
\newblock {\em Pharmaceutics}, 12(4):297, 2020.

\bibitem{vrijens2008}
Bernard Vrijens, G{\"a}bor Vincze, Paulus Kristanto, John Urquhart, and Michel
  Burnier.
\newblock Adherence to prescribed antihypertensive drug treatments:
  longitudinal study of electronically compiled dosing histories.
\newblock {\em Bmj}, 336(7653):1114--1117, 2008.

\bibitem{burnier2013}
Michel Burnier, Gregoire Wuerzner, Harry Struijker-Boudier, and John Urquhart.
\newblock Measuring, analyzing, and managing drug adherence in resistant
  hypertension.
\newblock {\em Hypertension}, 62(2):218--225, 2013.

\bibitem{gidal2021}
Barry~E Gidal, Jim Ferry, Larisa Reyderman, and Jesus~E Pi{\~n}a-Garza.
\newblock Use of extended-release and immediate-release anti-seizure
  medications with a long half-life to improve adherence in epilepsy: A guide
  for clinicians.
\newblock {\em Epilepsy \& Behavior}, 120:107993, 2021.

\bibitem{wheless2018}
James~W Wheless and Stephanie~J Phelps.
\newblock A clinician's guide to oral extended-release drug delivery systems in
  epilepsy.
\newblock {\em The Journal of Pediatric Pharmacology and Therapeutics},
  23(4):277--292, 2018.

\bibitem{vadivelu2011}
Nalini Vadivelu, Alexander Timchenko, Yili Huang, and Raymond Sinatra.
\newblock Tapentadol extended-release for treatment of chronic pain: a review.
\newblock {\em Journal of pain research}, 4:211, 2011.

\bibitem{norris1998}
J.R. Norris.
\newblock {\em {Markov Chains}}.
\newblock Statistical {\&} Probabilistic Mathematics. Cambridge University
  Press, 1998.

\bibitem{billingsley2013}
P~Billingsley.
\newblock {\em Convergence of probability measures}.
\newblock John Wiley \& Sons, 2013.

\bibitem{durrett2019}
R~Durrett.
\newblock {\em Probability: {T}heory and examples}.
\newblock Cambridge university press, 2019.

\end{thebibliography}
\bibliographystyle{unsrt}

\end{document}